\documentclass[final,1p,times,authoryear]{elsarticle}
\usepackage{lineno,hyperref}
\usepackage{breakcites,xcolor}

\usepackage{amsthm}
\usepackage{amssymb}
\usepackage{amsmath}
\usepackage{amsfonts}
\usepackage{bm} 
\usepackage{mathrsfs} 
\usepackage{multicol}
\usepackage{algorithm}
\usepackage{enumerate}
\usepackage{comment}

\theoremstyle{definition}
\newtheorem{definition}{Definition}[section]

\newtheorem{remark}[definition]{Remark}
\newtheorem{example}[definition]{Example}

\theoremstyle{plain}
\newtheorem{lemma}[definition]{Lemma}
\newtheorem{proposition}[definition]{Proposition}
\newtheorem{theorem}[definition]{Theorem}
\newtheorem{corollary}[definition]{Corollary}

\def\rank{{\rm rank}}
\def\Iso{{\rm Iso}}
\def\crit{{\rm crit}}

\def\by{{y}}
\def\xb{\mathbf{x}}

\def\ab{{\bf a}}

\def\bx{{ x}}
\def\lV{{\tilde{V}}}
\def\bp{{ p}}
\def\bq{{q}}

\def\Cc{{\mathcal C}}

\def\LL{\mathcal{L}}

\def\bz{{z}}

\def\e{{\bf e}}

\def\KK{{\mathbb K}}

\def\NN{{\tt N}}

\def\D{{\mathcal D}}

\def\V{{ V}}

\def\L{{\mathcal L}}

\def\C{{\mathbb C}}
\def\PP{{\mathbb P}}
\def\N{{\mathbb N}}
\def\Q{{\mathbb Q}}
\def\R{{\mathbb R}}
\def\Z{{\mathbb Z}}

\def\Sing{{\rm Sing}}

\journal{Journal of Symbolic Computation} 
\begin{document}

\begin{frontmatter}
\title{Smooth Points on Semi-algebraic Sets}

\author[address1]{Katherine Harris}
\ead{harriske@beloit.edu}
\author[address2]{Jonathan D. Hauenstein}
\ead{hauenstein@nd.edu}
\author[address3]{Agnes Szanto}
\ead{aszanto@ncsu.edu}

\address[address1]{Department of Mathematics and Computer Science, Beloit College, \\
700 College Avenue, Beloit, Wisconsin, 53511, USA.}

\address[address2]{Department of Applied and Computational Mathematics and Statistics, University of Notre Dame, \\
102G Crowley Hall, Notre Dame, Indiana, 46556, USA.} 

\address[address3]{Department of Mathematics, North Carolina State University, \\
Campus Box 8205, Raleigh, North Carolina, 27695, USA.}

\begin{abstract}
Many algorithms for determining properties of real algebraic or semi-algebraic sets rely upon the ability to compute smooth points.  In this paper, we present a procedure based on computing the critical points of some well-chosen function that guarantees the computation of smooth points in each bounded connected component of a (real) atomic semi-algebraic set. Our technique is intuitive in principal, performs well on previously difficult examples, and is straightforward to implement using existing numerical algebraic geometry software.  The practical efficiency of our approach is demonstrated by solving a conjecture on the number of equilibria of the Kuramoto model for the $n=4$ case. We also apply our method to design an algorithm to compute the real dimension of algebraic sets, the original motivation for this research. We compare the efficiency of our method to existing methods to compute the real dimension on a benchmark family.
\end{abstract}


\end{frontmatter}

\section{Introduction}


When studying an atomic semi-algebraic set $S\subset\R^n$, i.e.,
{ \begin{eqnarray}\label{eq:S}
S=\left\{x\in \R^n : f_1(x)=\cdots=f_{s}(x)=0,\; q_1(x)>0, \ldots, q_{m}(x)>0\right\}\end{eqnarray}}for some polynomials
$f_1, \ldots, f_s, q_1, \ldots, q_m\in \R[x_1, \ldots, x_n]$, one often first considers the complex variety $V=\{x\in \C^n\;:\; f_1(x)=\cdots=f_{s}(x)=0\}$, known as its {\em algebraic relaxation}, and deduces~properties of $S$ from the properties of $V$. In particular, if $S$ contains a smooth point and~$V$ is  irreducible, then $S$ is Zariski dense in $V$, and so all of the algebraic information of $S$ can be determined from~$V$. Thus, deciding the existence of smooth points in semi-algebraic sets and finding such points is a central problem in real algebraic geometry with many applications. 
An example of such application is computing all typical ranks of tensors \cite{RealTypicalRanks2018,Kruskal1989,TypicalRank2012}.

One of the main results of this paper is to give a new technique to compute smooth points on bounded connected components of atomic semi-algebraic sets. When $V$ is equidimensional, our method is simple and suggests a natural implementation using numerical homotopy methods. It extends other approaches that compute  sample points on   real semi-algebraic sets, such as computing  the critical points of the distance function,  in the sense that our method also guarantees the smoothness of the sample points. We illustrate this advantage on ``Thom's lips,'' 
in which critical points of the distance function are often at the singularities~\cite[Ex.~2.3]{WuReid13}, while our method always computes smooth points. 

The main idea is straightforward when $V$ is irreducible. If a polynomial $g$  vanishes on the singular points of~$V$, but does not vanish on all of~$V$, then the set of 
extreme points of $g$ on $V\cap \R^n$ must contain 
a smooth point on every bounded connected component of $V\cap \R^n$, if such points exist. We extend this idea beyond the case when $V$ is irreducible and $S=V\cap \R^n$ to the general case. We handle the case when   $V$ is  reducible and its irreducible components have different dimensions  using {\em infinitesimal deformations} of $V$ and limits. We show that this limiting approach is well-suited for numerical homotopy continuation methods after we  translate an infinitesimal real deformation (that may only work for arbitrary small values) into a complex deformation that works 
along a real arc parameterized by 
the interval~$(0,1]$. Finally, we present a novel technique to compute the required~$g$ polynomial using {\em deflations}, and compare its degree bounds to traditional symbolic approaches (see Proposition \ref{prop:boundg}). In fact, Corollary \ref{cor:boundLe} proves that our {\sc Real Smooth Point Algorithm} performs well if the depth of the deflations (i.e., the number of iterations)  is small. 

To demonstrate the practical efficiency of our new approach, we present the solution of a conjecture for the first time: 
counting the equilibria of the Kuramoto model in the $n=4$ case given in \cite{KuramotoConjecture} (see \cite{Kuramoto1975} for the original model and  \cite{Owenetal2018} for a detailed historical overview and additional references).

We also apply our method to compute the real dimension of a semi-algebraic set. The difficulty of this problem, compared to its complex counterpart, is that in many cases the real part of a semi-algebraic set lies within the singular set of the complex variety containing it so that its real dimension is smaller than the complex one. In terms of worst case complexity bounds of the existing algorithms in the literature, it is an open problem if the real dimension can be computed within the same asymptotic complexity bounds as the complex dimension. The original
motivation for this research was to try to find an algorithm for the real dimension that has worst case complexity comparable to its complex counterpart. 
Even though this paper is presented using computational tools
from {\em numerical algebraic geometry} (c.f.,~ \cite{SommeseWampler2005,BertiniBook2013}), 
all procedures can be translated to symbolic methods for polynomials with rational coefficients. In fact, 
after performing a worst case complexity estimate for a symbolic version, we unfortunately found that it does not improve the existing complexity bounds in the worst case (see \cite{LairezSafey2021} and the references therein). This is one of the reasons we present our results in a numerical algebraic geometry setting and give evidence of the efficiency with implementation on a benchmark family.  As mentioned above, in Proposition~\ref{prop:boundg} and Corollary~\ref{cor:boundLe}, we give bounds on the degrees of the polynomials appearing in our algorithms and the number of homotopy paths they follow, 
highlighting the advantages and disadvantages of our approach compared to other purely symbolic~techniques.

\subsection{Related Work}

There are many approaches in the literature to compute at least one real point on every connected component of a semi-algebraic set. Methods using projections to obtain a cell decomposition based on sign conditions go back to Collins' Cylindrical Algebraic Decomposition (CAD) algorithm described in \cite{Collins}. Improved symbolic methods using critical points or generalized critical points of functions along with infinitesimals and randomization can be found in \cite{RouRoySaf2000,AuRouSaf02,Safey2007,Faugereetal2008}. The current state of the art deterministic symbolic algorithm is given in \cite[Alg. 13.3]{BasuPollackRoyBook} which computes sample points on each connected component of all realizable sign conditions of a polynomial system and gives a complexity analysis. The most recent application of this technique is in \cite{SafYanZhi2018,SafYanZhi2019} where the authors compute smooth points on real algebraic sets in order to compute the real radical of polynomial systems and analyze complexity.

Another line of work has been developed in parallel focuses on computing critical points while utilizing the tool of \textit{polar varieties}, introduced and developed in \cite{Banketal1997,SafeySchost03,Banketal04,Banketal09,Banketal2010,Banketal14,SafeySpaen2016}. Alternatively, a homotopy-based approach computing the critical points of the distance function from a generic point or a line is presented in~\cite{Hauenstein2013,WuReid13}.
 It is important to note, however, that all of these methods only guarantee the finding at least one real point on every connected component of a semi-algebraic set, rather than real \textit{smooth} points.

The real dimension problem has similarly been widely studied and
the current state of the art deterministic algorithm is given by  \cite[Alg. 14.10]{BasuPollackRoyBook} computing all realizable sign conditions of a polynomial system. This approach improves on previous work in \cite{Vorobjov1999} to obtain a complexity result with a better dependence on the number of polynomials in the input by utilizing a block elimination technique first proposed in \cite{GriVor88}.
More recent work has been presented giving probabilistic algorithms utilizing polar varieties which improve on complexity bounds even further in \cite{MohabElias2013,BanSaf2015}. Very recently,  \cite{LairezSafey2021} gave a symbolic algorithm with improved complexity to compute the dimension of real algebraic sets  using ``level sets" and use them to reduce the problem
to one of lower dimension.  In the present paper, we use a benchmark family that appeared in  \cite{LairezSafey2021} to demonstrate the efficiency of our method. 


One can also compute the real dimension by computing the real radical of a semi-algebraic set, first studied in~\cite{BeckNeu93}
with improvements and implementations in \cite{Neuhaus98,Zeng1999,Spang08,Chenetal2013}. The most recent implementation can be found in \cite{SafYanZhi2018,SafYanZhi2019} as mentioned above. Their approach is shown to be efficient in the case when the algebraic set is smooth, but the iterative computation of singularities can increase the complexity significantly in the worst case. An alternative  method using semidefinite programming techniques was proposed by \cite{Wang2016,Maetal2016}. 

 \section{Preliminaries}

\subsection{Basic Definitions}

The following collects some basic
notions used throughout the paper,
including atomic semi-algebraic sets,
semi-algebraic sets, and real algebraic
sets.

A set $S\subset \R^n$ is an {\em atomic semi-algebraic set} if it is of the form of~\eqref{eq:S}.
A set $T\subset \R^n$ is a {\em semi-algebraic set}  if it is a finite  union of {atomic semi-algebraic sets}. 
A set $U\subset \R^n$  is a {\em real algebraic set} if it is defined by polynomial equations only.

Smoothness on atomic semi-algebraic sets
is described next.

\begin{definition}
Let $S\subset \R^n$ be an atomic semi-algebraic set as in~\eqref{eq:S}. A point $\bz\in S$ is {\em smooth} (or nonsingular) in $S$ if $z$
is smooth in the algebraic set
$$V(f_1, \ldots, f_s) =\{x\in\C^n
\; :\;f_1(x) = \cdots = f_s(x) = 0\},$$ 
i.e., if
there exists a unique irreducible component $V\subset V(f_1, \ldots, f_s)$ containing~$\bz$ such that
$$
\dim T_\bz(V)= \dim V 
$$
where $T_\bz(V)$ is the tangent space of $V$ at $\bz$. We denote by $\Sing(S)$ the set of singular (or non-smooth)  points in $S$.
\end{definition}

An algebraic set $V\subset \C^n$ is {\em equidimensional} of dimension $d$ if every irreducible component of~$V$ has dimension $d$.  The following defines the (local) real dimension of semi-algebraic sets from \cite[\S 5.3]{BasuPollackRoyBook}.

\begin{definition}\label{def:realdimension}
For a semi-algebraic set $S\subset \R^n$,
its {\em real dimension} $\dim_\R S$ is the largest $k$ such that there exists an injective semi-algebraic map from $(0,1)^k$ to $S$.  Here, a map 
$\varphi:(0,1)^k\rightarrow S$ is semi-algebraic if the graph of $\varphi$ in $\R^{n+k}$ is semi-algebraic. By convention,  the dimension (real or complex) of the empty set is $-1$.
 \end{definition}
 
 \begin{definition}\label{def:reallocaldimension}
 Consider a point $\bz \in S \subset \R^n$, where $S$ is a semi-algebraic set. The {\em local real dimension} of $S$ at $\bz$ is the maximal real dimension of the closure of every connected component~$C_j$ of $S$ such that $\bz \in \overline{C_j}$. 
 \end{definition}
 
 

The main ingredient in our results is the following theorem that was proven in \cite[Theorem 12.6.1]{Marshall2008}.

\begin{theorem}\label{thm:realconvdim}
Let $V\subset \C^n$ be an irreducible algebraic set. Then
$$
\dim_\R \left(V\cap\R^n\right)=\dim_\C V
$$
if and only if there exists $\bz\in V\cap \R^n$ that is smooth. 
\end{theorem}

\subsection{Semi-algebraic to Algebraic}

In this subsection, we show that our problem on atomic semi-algebraic sets can be reformulated as a problem on real algebraic sets. This will allow us to use homotopy continuation methods for solving polynomial equations.

The following shows that smooth points on each connected component of an atomic semi-algebraic set $S$ can be obtained as  projections of smooth points of some  real algebraic set.  

\begin{proposition}\label{prop:semalgalg}
Let $S$ be an atomic semi-algebraic set as in \eqref{eq:S} and
\begin{eqnarray*}W:= \left\{(x, z)\in \R^n\times\R^{m} \; :\; f_1(x)=\cdots=f_{s}(x)=0,\,z_1^2q_1(x)-1= \cdots= z_{m}^2q_{m}(x)-1=0\right\}.
\end{eqnarray*}
If $\by\in W$ is smooth, then $\pi_x(\by)\in S$ is also smooth. Conversely, if $\bx\in S$ is smooth, then 
$(x,z)$ is smooth in $W$ for all
$z=(z_1, \ldots, z_m)\in \R^m$ such that $(x,z)\in W$.
\end{proposition}

\begin{proof}
Without loss of generality, we can assume that
 $f_1, \ldots, f_s$ generate a prime ideal. The Jacobian matrix of the polynomial system defining $W$ has the block structure
$$
J(\bx,\bz)=\begin{array}{|c|c|}
\hline
Jf(\bx)& 0\\
\hline
* & {\rm diag}(2z_iq_i(\bx))\\
\hline
\end{array}
$$
Since  for $(\bx,\bz)\in W$ we have $z_ig_i(\bx)\neq 0$,  the Jacobian matrix $Jf(\bx)$ has full column rank if and only if $J(\bx,\bz)$ has full column rank, which proves the claim.
\end{proof}

Therefore, for the rest of the paper, we assume that we are given a {\em real algebraic set} and the goal is to compute smooth points on each connected component.

\begin{remark}
Applying Proposition \ref{prop:semalgalg} to reduce to the case of algebraic sets may not be the only possibility. We give this procedure for simplicity of presentation and implementation, but it is possible that another symbolically equivalent reduction method would produce a more efficient numerical algorithm. While a rigorous analysis of this falls outside the scope of this paper, we suggest the reader explores this more if they are studying an example or application that requires more computational efficiency. 
\end{remark}

 \subsection{Boundedness} 

The next reduction is to replace
an arbitrary real algebraic set 
with a compact one. 

\begin{proposition}\label{prop:bounded}
Let $f_1, \ldots, f_{s-1}\in \R[x_1, \ldots, x_{n-1}]$ and consider \mbox{$\bp=(p_1, \ldots, p_{n-1}) \in\R^{n-1}$}. 
Let \mbox{$\delta\in\R_+$},
introduce a new variable~$x_n$,
and consider
$$
f_{s}:=(x_1-p_1)^2+\cdots + (x_{n-1}-p_{n-1})^2+x_n^2-\delta.$$
Then, $V(f_1, \ldots, f_{s})\cap \R^n$ is bounded and 
{ $$
\pi_{n-1}\left(V(f_1, \ldots, f_{s})\cap \R^n\right)= V(f_1, \ldots, f_{s-1})\cap \left\{\bz\in \R^{n-1}\;:\; \|\bz-\bp\|^2\leq \delta\right\}
$$}where $\pi_{n-1}(x_1, \ldots, x_n)= (x_1, \ldots, x_{n-1})$.
\end{proposition}

\begin{remark}
The definition of $f_{s}$ above is based on a standard trick used in real algebraic geometry to make an arbitrary real algebraic set bounded (e.g., see \cite{BasuPolackRoy2006}). In general, $V\cap \R^{n-1}$ is embedded into a sphere in $\R^n$ around the origin of radius $1/\zeta$ where $\zeta$ is infinitesimal. Since numerical homotopy continuation methods are incompatible with infinitesimal variables,  in this paper we are only interested in computing points with bounded coordinates, so it is sufficient to embed its intersection  with a closed ball around $\bp$ of 
radius~$\sqrt{\delta}$ for some fixed $\delta\in \R_+$. In particular, we will not use infinitesimal variables in our algorithms.
\end{remark}
 
Later in the paper, when we assume that 
$V(f_1, \ldots, f_{s})\cap \R^n$ is bounded, we assume that we have applied Proposition \ref{prop:bounded} if necessary.

\subsection{Genericity Assumptions} The algorithms described in this paper make assumptions that certain points, matrices,  or linear polynomials are generically chosen from a vector space (over $\Q$, $\R$ or $\C$). In all of these cases, there exists a proper Zariski closed subset of the corresponding vector space  such that all choices outside this set yield correct results.
Therefore, a generic choice means 
it is outside of this proper Zariski closed subset. 
For algorithms which depend on generic choices, we follow the convention from the literature that they
compute the correct solutions with {\em algebraic probability one}~\cite[Chap.~4]{SommeseWampler2005}.  Effective
probability bounds can be obtained from bounds on
the degrees of the proper Zariski closed sets containing
the ``bad'' choices.   See \cite[Prop.~4.5]{KrPaSo2001} and \cite{EllSch2019} for such bounds for linear changes of variables for Noetherian position and transversality, respectively.  

\subsection{Witness Sets}

In this subsection, we  discuss  some main ideas from numerical algebraic geometry following \cite{BertiniBook2013}. In particular, we consider positive-dimensional algebraic sets and utilize a data structure for them that allows computation using classical homotopy continuation methods for square non-singular systems.  The key is the notion of {\em witness sets}, which will rely on the idea of slicing an algebraic set with a generic linear space.

\begin{definition}\label{def:witnessset} 
If an algebraic set $V\subset\C^n$ is
equidimensional with $\dim(V) = k$, a {\em witness set} for~$V$ is the triple $(F,L,W)$ such that
\begin{itemize}
\item $F\subset \C[\bx]$ is a {\em witness system} for $V$, i.e., each irreducible component of $\V$ is an irreducible component of $\V(F)$,
\item $L\subset \C[\bx]$ is a {\em linear system} where $\V(L)$ is a linear space of codimension $k$ that intersects~$V$ transversely, and
\item $W\subset \C^n $ is a {\em witness point set} which is equal to $V\cap \V(L)$.
\end{itemize}
\end{definition}

We note that the number of points in the witness point set $W$ in the above definition will determine the number of paths we need to track with homotopy continuation methods, directly impacting  the complexity of the numerical algebraic geometry computations. Note also that this number is an invariant of $V$ called its {\em degree}. 

Although the witness point set provides some information,
a witness system is needed to perform any additional computations, such as deciding whether a given point lies on the algebraic set $V$.  {\sc Membership Test Algorithm \ref{Alg:MT}} follows the approach of \cite[Section 8.4]{BertiniBook2013} in order to do this and is correct with algebraic probability one.

\begin{algorithm}[H]
\caption[Alg:MT]{\sc MembershipTest}
\label{Alg:MT}
\begin{description}
\item[Input:] $\bp \in \C^n$ and $(F,L,W)$ a witness set for some equidimensional algebraic set $V \subset \C^n$. 
\item[Output:] TRUE if $\bp \in V$ and FALSE if $\bp\not\in V$.
\end{description}
\begin{enumerate} 
\item Choose generic linear polynomials $L'$ with $\bp \in \V(L')$ and $\dim(\V(L')) = \dim(\V(L))$.
\item 
$H(\bx,t) := 
\left[  
F(\bx), tL(\bx)+(1-t)L'(\bx)
\right].
$
\item  Track the finitely many homotopy paths of $H(\bx,t)$ starting from the witness point set $W$ for $t=1$, 
obtaining the witness point set $W':=V\cap \V(L')$ at $t=0$.

\item If $\bp \in W'$, return TRUE. Else, return FALSE.   
\end{enumerate}
\end{algorithm}


\subsection{Isosingular Deflation}\label{sec:Iso}

As we mentioned in the Introduction, one of the ingredients of our algorithm for computing smooth points on a real algebraic set $V\cap \R^n$ is a polynomial $g$ that vanishes on the singular points of $V$, but does not vanish identically on the irreducible components of $V$. We give an algorithm to compute such a $g$ in Section \ref{sec:compg} using {\em isosingular deflation}. This subsection summarizes the basic definitions and results that we use  in Sections \ref{sec:defwit} and  \ref{sec:compg}. Further details on isosingular deflation can be found in \cite{HauensteinWampler2013}. 

\begin{definition}\label{Def:IDO}
Let $f_1, \hdots, f_s \in \C[\bx], F_0= \{f_1, \hdots, f_s\}$, and $\bz\in \V(F_0)\subset \C^n$.
The {\em isosingular deflation operator}~$\D$ is defined via
$$(F_1,\bz) := \D(F_0,\bz)$$
where $F_1\subset \C[\bx]$ consists of $F_0$ and all $(r+1)\times (r+1)$
minors of the Jacobian matrix $JF_0$ for~$F_0$ 
where $r = \rank~JF_0(\bz)$.  Thus, $\bz\in \V(F_1)$,
meaning that we can iterate this operator to construct a 
sequence of systems $F_{j}\subset \C[\bx]$ 
with $(F_{j},\bz) = \D(F_{j-1},\bz)=\D^j(F_0, \bz)$~for~\mbox{$j\geq 1$}.

We say that $F\subset \C[\bx]$ is the {\em isosingular deflation} of $F_0$ at $\bz$ if there exists a minimal $j\geq 0$ such that $(F, \bz)=\D^j(F_0, \bz)$ and $\dim {\rm NullSpace} (JF(\bz))=\dim_{F}(\bz)$, where 
 $\dim_{F}(\bz)$  is the maximal dimension of the irreducible components of $\V(F)$ containing $\bz$ (called the {\em local dimension} of $\bz$ with respect to $F$). 
 \end{definition}
 
 To compute the isosingular deflation of $F_0$ at $\bz$ we refer to   \cite[Algorithm 6.3]{HauensteinWampler2013}.

Using the deflation operator, we can now formally define the isosingular sets and singular points of our algebraic set in this context. 

\begin{definition}\label{def:iso}
Let $f_1, \hdots, f_s \in \C[\bx], F_0= \{f_1, \hdots, f_s\}$, and $\bz\in \V(F_0)\subset \C^n$. Let $\D$ be the isosingular deflation operator defined in Definition \ref{Def:IDO}. We define 
\begin{itemize}
\item 
The {\em deflation sequence } of $F_0$ at $\bz$ is $\{d_k(F_0, \bz)\}_{k=0}^\infty$ where 
$$d_k(F_0, \bz)={\rm dnull}(F_k, \bz):=\dim{\rm NullSpace} JF_k(\bz)$$ with $JF_k$ the Jacobian matrix of $F_k$ with $(F_k,\bz)=\D^k(F_0,\bz)$. 

\item 
Let $V\subset \V(F_0)$ be a non-empty  irreducible algebraic set. Then
$V$ is an   {\em isosingular set} of $F_0$ if there exists a sequence  $\{c_k\}_{k=1}^\infty$ such that $V$ is an irreducible component of 
$$\overline{\{\bz\in \V(F_0) \;:\; d_k(F_0, \bz)=c_k, k\in \N\}}.$$

\item Let $V\subset \V(F_0)$ be a non-empty  irreducible algebraic set. Then
$\Iso_{F_0}(V)$ is the unique isosingular set with respect to  $F_0$ containing $V$ such that  $\Iso_{F_0}(V)$ and $V$ have the same deflation sequence with respect to $F_0$.

\item 
Let $V$ be an isosingular set for $F_0$. The set of {\em singular points of $V$} with respect to $F_0$ is
$$\Sing_{F_0}(V)=\left\{\bz\in V\;:\;\{d_k(F_0, \bz)\}_{k=0}^\infty\neq \{d_k(F_0, V)\}_{k=0}^\infty\right\}.$$
Here, $d_k(F_0, V)$ is meant for a generic point in $V$.

\item The {\em local dimension} of $\bz$ with respect to $F_0$, denoted by $\dim_{F_0}(\bz)$,  is the maximal dimension of the irreducible components of $\V(F_0)$ containing $\bz$. 
\end{itemize}
\end{definition}

We next detail some particular results on isosingular sets which will be important to our methods going forward. 
 The following theorem states that the singular points of an algebraic set are preserved under isosingular deflation. We use this result in the proof of Theorem \ref{thm:gcorrect}.

\begin{theorem}\label{thm:singdeflation} \cite[Theorem 5.9]{HauensteinWampler2013}
Let $V$ be an isosingular set for $F_0$ as in Definition \ref{def:iso}. Then if $\bz \in V$ and $\bz \in \Sing(\V(F_0))$ then  $\bz \in \Sing_{F_0}(V)$. 
\end{theorem}

Finally, we have the following theorem which gives an isosingular deflation approach for constructing witness sets of the intersection of a known witness set with another algebraic set. We use this result in the proof of Theorem \ref{thm:deflcorrect}.

\begin{theorem}\label{thm:isointersect} \cite[Theorem 6.2]{HauensteinWampler2017}
Given $g_1, \hdots, g_r \in \C[\bx]$, let $Z$ be a union of irreducible components of $\V(g_1, \hdots, g_r)$. Suppose $f_1, \hdots, f_s \in \C[\by]$, $F(\bx,\by) = \{g_1(\bx), \hdots, g_r(\bx), f_1(\by), \hdots, f_s(\by)\}, \Delta=\{(\bx,\bx) : \bx \in \C^n\},$ and $\pi(\bx,\by) = \bx$. If $A \subset Z \cap \V(f_1, \hdots, f_s)$ is an irreducible component, then there exists a nonempty Zariski open set $U \subset A$ such that for all $\bp \in U$, $A$ is an irreducible component of $\;\pi\bigg(\Iso_F((\bp,\bp)) \cap \Delta\bigg)$.  
\end{theorem}
We provide the following illustrative example for the theorem. 

\begin{example}
Let $g(x,y,z):=(x+y+z)y$ defining a witness system for  $Z:= \V(x+y+z)$. Let $f(x,y,z):=y$ and note that $ Z \cap \V(f) = \V(x+z,y)$ is irreducible. We construct 
$$F_0(x,y,z,x',y',z') = [g(x,y,z)=(x+y+z)y,f(x',y',z')=y'].$$
Choose a generic witness point $\bp = (a, 0, -a) \in  A=Z\cap V(f)$ for some fixed $a \in \C.$ Then we compute the deflation sequence of $F_0$ at $(a,0,-a, a, 0,-a)$   as $5,3,3,\hdots$ such that 
$$\Iso_{F_0}((\bp,\bp)) = \{(b,0,-b,c,0,d) : b,c,d \in \C\}.$$
The polynomial system defining this 3-dimensional isosingular set is given by adding the $2 \times 2$ minors of the Jacobian of $F_0$ to $F_0$, giving 
\[
F_1(x,y,z,x',y',z') = \begin{bmatrix} 
                    (x+y+z)y \\ 
                    y' \\
                    y \\
                    x + 2y + z \\
                    y 
                    \end{bmatrix}.
\]
By Theorem \ref{thm:isointersect}, $Z\cap V(f)$ is an irreducible component of $\V(F_1(x,y,z,x,y,z))$ and $G(x,y,z):=F_1(x,y,z,x,y,z)$ suffices as a witness system for $Z\cap V(f)$. We note that in this example, removing the redundancies in $G$ would in fact show that $[x+z,y]$ is sufficient as a witness system for $Z\cap V(f)$. 
\end{example}

\section{Computation of Real Smooth Points -- Equidimensional Case}\label{sec:basiccase}

This section contains 
our main results for the special case when the complex algebraic variety is equidimensional. In subsequent secions we consider the general case, where we use deformations and limits of algebraic sets.   


\begin{theorem}\label{thm:noperturb}
Let $f_1, \ldots, f_{s}\in \R[x_1, \ldots, x_n]$ and assume that $V:=V(f_1, \ldots, f_{s})\subset \C^n$ 
is equidimensional of dimension $n-s$. 
Suppose that $g\in \R[x_1, \ldots, x_n]$ 
satisfies the following conditions:
\begin{enumerate}
    \item $\Sing(V)\cap \R^n\subset V(g)$;
    \item $\dim \left(V\cap V(g)\right)<n-s$.
\end{enumerate}    
Then the set of points where $g$ restricted to $V\cap\R^n$ attains its extreme values intersects
each bounded connected component of $\,\left(V\setminus \Sing(V)\right)\cap \R^n$.
\end{theorem}

The proof of this theorem is based on the following lemma.

\begin{lemma}
\label{lem:noperturb}
Let $V$ be as in Theorem \ref{thm:noperturb}.
  Let $g\in \R[x_1, \ldots, x_n]$  such that
$
\dim \left(V\cap V(g)\right)<n-s.
$
    Then, either $\,\left(V\setminus V(g)\right)\cap \R^n = \emptyset$ or $g$ restricted to $V\cap \R^n$ 
    attains a non-zero extreme value on each bounded connected component of  $\,\left(V\setminus V(g)\right)\cap \R^n$.
\end{lemma}

\begin{proof} Assume that  $\,\left(V\setminus V(g)\right)\cap \R^n \neq  \emptyset$ and let $C$ be a bounded connected component of
the set $(V\setminus  V(g))\cap \R^n$.
Since  $C \not\subset V (g)$, there exists  $ x \in C$  with $g(x) \neq 0$. Let $\overline{C}$  be the Euclidean
closure of $C$ so that $\overline{C}\subset V \cap \R^n$ is closed and bounded,
and $g$ vanishes identically on $\overline{C}\setminus C$. By the extreme value theorem,
$g$ attains both a minimum and a maximum on $\overline{C}$. Since $g$ is not identically zero on $\overline{C}$,
either the minimum or the maximum value of $g$ on $\overline{C}$ must be nonzero, so~$g$
attains a non-zero extreme value on $C$. 
\end{proof}

\begin{proof}[Proof of Theorem~\ref{thm:noperturb}]
Assume that  $\,\left(V\setminus \Sing(V)\right)\cap \R^n \neq  \emptyset$. By Theorem \ref{thm:realconvdim}, $\dim_\R V\cap \R^n =n-s$.  By (2), $\,\left(V\setminus V(g)\right)\cap \R^n \neq  \emptyset$. 
By~(1), $\,\left(V\setminus V(g)\right)\cap \R^n \subset\left(V\setminus \Sing(V)\right)\cap \R^n $,
so the bounded connected components of $\left(V\setminus V(g) \right)\cap \R^n $ are subsets of the bounded connected components of $\,\left(V\setminus \Sing(V)\right)\cap \R^n$. By Lemma \ref{lem:noperturb},~$g$ restricted to $V\cap \R^n$ attains a non-zero extreme value on each bounded connected component of  $\,\left(V\setminus V(g)\right)\cap \R^n$, thus yielding a point in every bounded connected component of $\left(V\setminus \Sing(V)\right)\cap \R^n $.
\end{proof}

{\sc Algorithm \ref{alg:RSP2}} in Section \ref{sec:smooth2}
computes real smooth points when $V(f_1, \ldots, f_s)$ is not equidimensional by using deformations and limits. 
However, the same algorithm can be used in the equidimensional case with input $f_1, \ldots, f_s$ and $\ab=0\in \R^s$, i.e~without~deformation. 

\begin{figure}[!hb]
    \centering
    \includegraphics[scale=0.3]{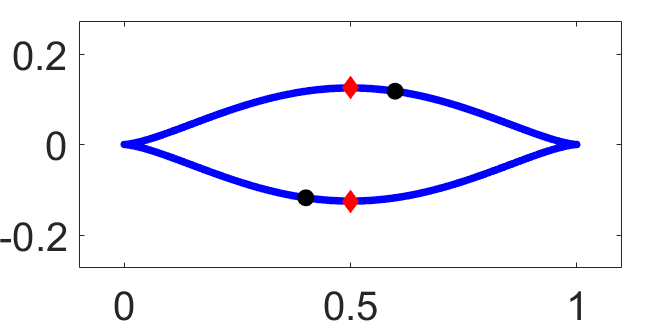}
    \caption{``Thom's lips"}\label{fig:Thom}
\end{figure}

\begin{example}
An example of a real curve with two singular cusps
called ``Thom's lips" defined by $f = y^2-(x(1-x))^3$
is shown in Figure~\ref{fig:Thom}.
An obvious choice of $g$ which satisfies the conditions of Theorem~\ref{thm:noperturb} is $g=x(1-x)$. Using Lagrange multipliers to optimize with respect to $g$ results in two points $(0.5,\pm 0.125)$ plotted as red diamonds.
Alternatively, the polynomial~$g$ can be constructed 
algorithmically (see Section~\ref{sec:compg}) 
yielding, e.g., \mbox{$g=3(2x-1)(x(1-x))^2+2y$}
which produces two points plotted as black circles,
approximately $(0.5987,0.1178)$
and $(0.4013,-0.1178)$.
Both yield a
real smooth point on each
of the two connected components 
of \mbox{$\left(V \setminus \Sing(V)\right)\cap\R^n$}. 
We note that the first choice of $g$ demonstrates that when \Sing(V) is 0-dimensional, defining $g$ as a product of a coordinate of these points is a simple way to satisfy the conditions of Theorem~\ref{thm:noperturb}. The second choice of $g$ demonstrates the general method described in Section~\ref{sec:compg} which works in every dimension.
\end{example}

\section{Application to Kuramoto model}

The Kuramoto model from \cite{Kuramoto1975} is a dynamical system 
used to model synchronization amongst $n$ coupled oscillators.  
The maximum number of equilibria (i.e. real solutions to steady-state equations) 
for $n\geq4$ remains an open problem with details discussed in \cite{Owenetal2018}.
The following confirms the conjecture in \cite{KuramotoConjecture}
for $n = 4$.

\begin{theorem}\label{thm:Kuramoto}
The maximum number of equilibria for
the Kuramoto model with $n = 4$ oscillators is $10$.
\end{theorem}

The steady-state equations for the $n =4$ Kuramoto model
are 
$$\mbox{$f_i(\theta;\omega) = \omega_i - \frac{1}{4}\sum_{j=1}^4 \sin(\theta_i-\theta_j) = 0, \hbox{~for~}i = 1,\dots,4$}$$
parameterized by the natural frequencies $\omega_i\in\R$.  
Since only the angle differences matter, one can assume $\theta_4 = 0$
and observe a necessary condition for equilibria is 
$$0=f_1+f_2+f_3+f_4=\omega_1+\omega_2+\omega_3+\omega_4,$$
i.e., assume $\omega_4 = -(\omega_1+\omega_2+\omega_3)$.
Substituting $s_i=\sin(\theta_i)$ and $c_i = \cos(\theta_i)$ yields 
$$
\mbox{$
F(s,c;\omega) = \left\{\omega_i - \frac{1}{4} \sum_{j=1}^4 (s_ic_j - s_jc_i), s_i^2 + c_i^2 - 1, \hbox{~for~}i = 1,2,3\right\}$}$$
which is a polynomial system with variables $s = (s_1,s_2,s_3)$ and $c = (c_1,c_2,c_3)$,
parameters $\omega = (\omega_1,\omega_2,\omega_3)$,
and constants $s_4 = 0$ and $c_4 = 1$.

The goal is to compute the maximum number of isolated real solutions 
of $F = 0$ as $\omega$ varies over $\R^3$.  Let $D(\omega)$ be the discriminant polynomial of the system $F$, a polynomial in $\omega$ of degree~$48$.
The number of real solutions of $F$ is constant  in each connected component of $\R^3\setminus V(D)$. 
Since it is easy to see that there can 
be no real solutions if $|\omega_i|\geq \frac{n-1}{n}=0.75$, we need to compute at least one interior point
in each of the bounded connected components of $\R^3\setminus V(D)$.    
Applying Lemma \ref{lem:noperturb} with $f=0$ and $g=D$, i.e., by computing the real solutions of $\nabla D=0$ and $D\neq 0$,
accomplishes this task.  Exploiting symmetry 
and utilizing {\tt Bertini} (\cite{Bertini}), {\tt alphaCertified} (\cite{HauSot2012}),
and {\tt Macaulay2} (\cite{Macaulay2})  all solutions have been found and certified.
In fact, this computation showed that all real critical points of $D$ arose, up to symmetry, along
two slices shown in Figure~\ref{fig:Kuramoto}.
A similar computation then counted the number of real
solutions to $F = 0$ showing that the maximum number of equilibria is $10$.
All code used in these computations is available at \url{dx.doi.org/10.7274/r0-5c1t-jw53}.

\begin{figure}[!ht]
    \centering
    \[
    \begin{array}{cc}
    \includegraphics[scale=0.3]{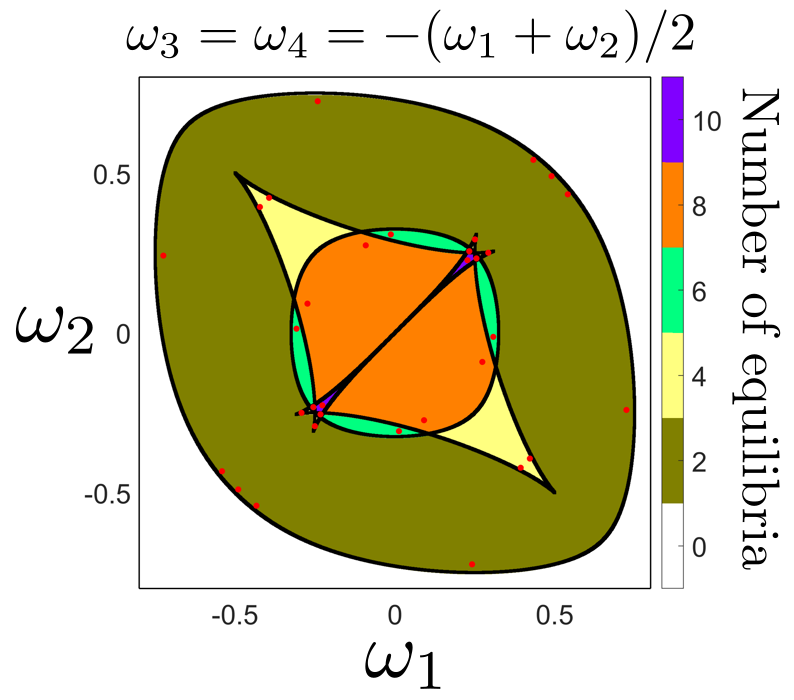} & \includegraphics[scale=0.3]{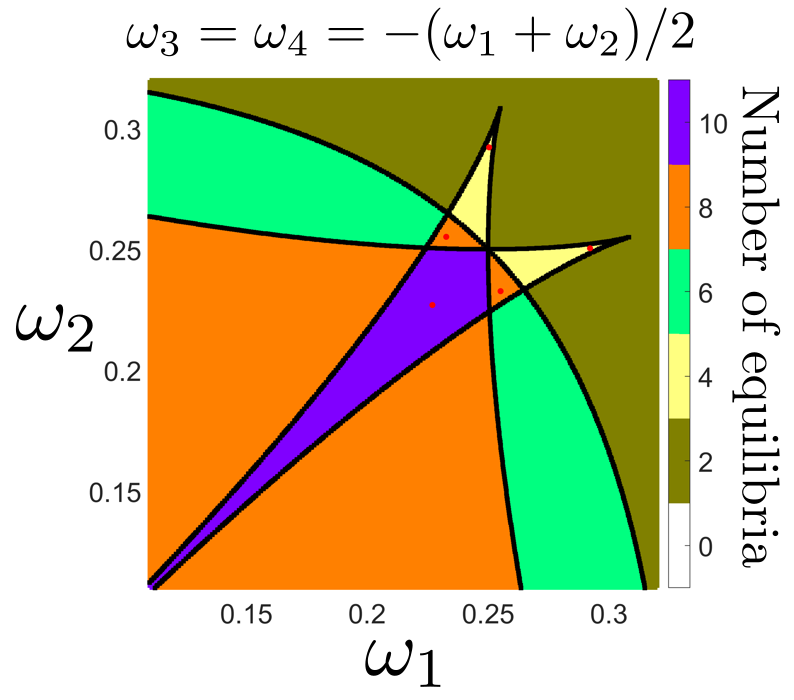} \\ 
    \hbox{(a)~one slice} & \hbox{(b)~zoomed in} \\ \\
    \includegraphics[scale=0.3]{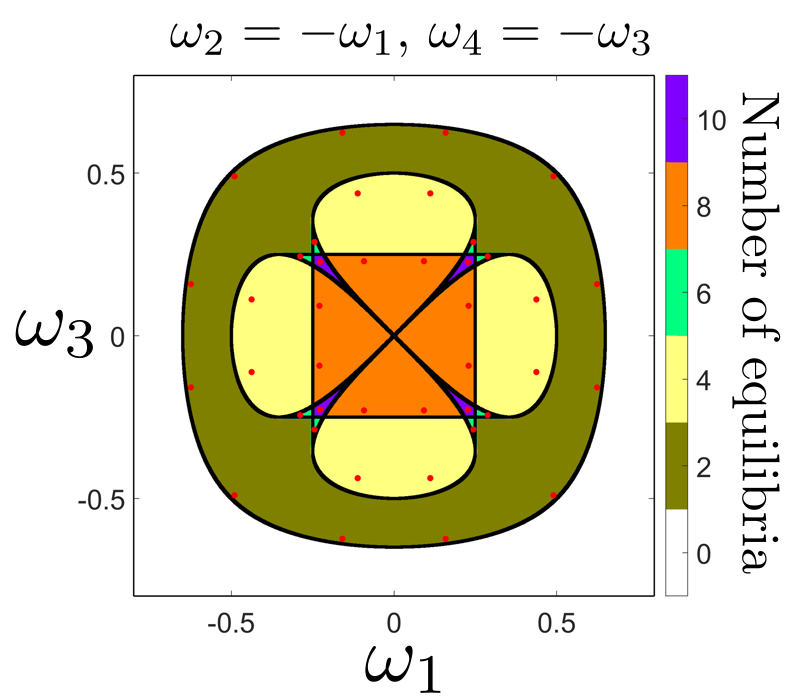} & \includegraphics[scale=0.3]{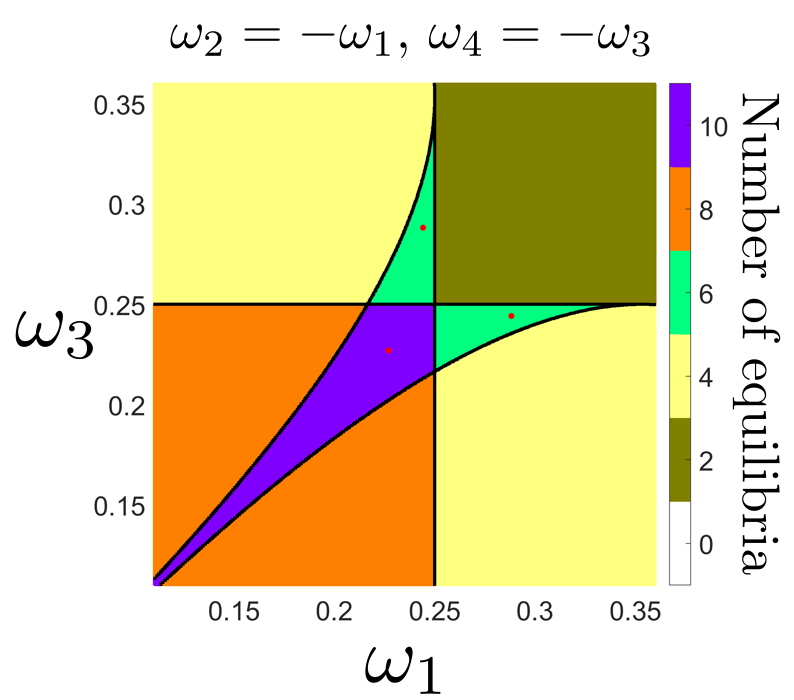} \\
    \hbox{(c)~other slice} & \hbox{(d)~zoomed in}
    \end{array} \]
    \caption{Compact connected regions and critical points for the Kuramoto model with $n = 4$}\label{fig:Kuramoto}
\end{figure}


\section{Perturbations and Limits of Real Hypersurfaces}\label{sec:RealPerturb}

In this section, we begin to construct the tools for an algorithm using the same ideas as in Section \ref{sec:basiccase}, but in full generality so the results will hold when $\V(f_1, \ldots, f_{s})$ is not equidimensional, i.e. there are some components of dimension greater than $n-s$. To do this, a standard approach (which can be seen in \cite{MohabElias2013} and \cite{BanSaf2015}) is to perturb the defining polynomials of the algebraic set by infinitesimals to obtain an equidimensional variety, which has a limit as the infinitesimals approach zero that is equidimensional and contains the $n-s$ dimensional components of $\V(f_1, \ldots, f_{s})$.



First, we need the definition of Puiseux series. 
\begin{definition}\label{def:Puiseux}
Let $\KK=\R$ or $\C$ and denote by $\KK\langle\varepsilon\rangle$ the {\em field of Puiseux series} over $\KK$, i.e.
$$
\KK\langle\varepsilon\rangle:=\left\{\sum_{i\geq i_0} a_i \varepsilon^{i/q}\;:\; i_0\in \Z, q\in \Z_{>0}, a_i\in\KK\right\}.
$$
A Puiseux series $z=\sum_{i\geq i_0} a_i \varepsilon^{i/q}\in \KK\langle\varepsilon\rangle$ is called {\em bounded} if $i_0\geq 0$.
\end{definition}

We note that the field of real Puiseux series $\R\langle \varepsilon \rangle$ is a real closed field, thus the intermediate value theorem holds (c.f. \cite[Theorem 2.11]{BasuPollackRoyBook}). 

We also note that for complex Puiseux series $\C\langle \varepsilon \rangle$, similar results hold  when we replace real closed with algebraically closed (c.f. \cite[Theorem 2.92]{BasuPolackRoy2006}). 

We establish a notation for the concept of an {\em extension} of a semi-algebraic set $S \subset \R^n$ to $\R \langle \varepsilon \rangle$. 

\begin{definition}\label{def:extension}
Given a semi-algebraic set $S \subset \R^n$, the {\em extension} of $S$ to $\R \langle \varepsilon \rangle^n$, denoted $\text{Ext}(S)$, is the semi-algebraic subset $S' \subset \R \langle \varepsilon \rangle^n$ defined by the same equations and inequalities as $S$, but considering their solutions in $\R \langle \varepsilon \rangle^n$. For a polynomial map $\varphi:S\rightarrow S'$ where $S\subset R^n$, $S'\subset \R^m$ semialgebraic sets, $\text{Ext}(\varphi) $ denotes the map $\varphi': \text{Ext}(S)\rightarrow \text{Ext}(S')$ defined by the same polynomials as $\phi$. 
\end{definition}

Assume $f_1, \ldots, f_s\in \R[\bx]$ and $\varepsilon >0$ a real infinitesimal. Let  $F = f_1^2 + \cdots + f_s^2$ and note that $\V(f_1, \hdots, f_s) \cap \R^n = \V(F) \cap \R^n$. Consider $V_{\varepsilon}:= \V(F- \varepsilon)\subset \C\langle \varepsilon \rangle$ as the perturbed version of the algebraic set $\V(F)$. One reason for us to perturb our algebraic set is to obtain smoothness. 

\begin{lemma}\cite[Lemma 3.5]{RouRoySaf2000}\label{lem:perturbsmooth}
$V_\varepsilon$ is a smooth hypersurface. 
\end{lemma}

 The following result \cite[Proposition 12.36]{BasuPollackRoyBook} states that semi-algebraicity is preserved as $\varepsilon$ limits to 0. 

\begin{lemma}\label{lemma:limsemi}
Let $S \subset \R\langle \varepsilon \rangle^n$ be a semi-algebraic set. Then $\lim_{\varepsilon \rightarrow 0}(S)$ is a closed semi-algebraic set. Furthermore, if $S$ is bounded and connected, then $\lim_{\varepsilon \rightarrow 0}(S)$ is connected. 
\end{lemma}

The next proposition on the limits of perturbed connected components of a real algebraic set appeared in the unpublished work \cite{MohabElias2013}. We restate and prove it here for completeness. 

\begin{proposition}\label{prop:connectedlimits}
Assume $f_1, \ldots, f_s\in \R[\bx], F = f_1^2 + \cdots f_s^2$, and $\varepsilon >0$ a real infinitesimal. Let  $V_\varepsilon :=\V(F-\varepsilon)\subset \C\langle \varepsilon \rangle^n$ and  $V:=\lim_{\varepsilon\rightarrow 0}V_\varepsilon \subset \C^n$. Suppose $C$ is a connected component of $V \cap \R^n.$ Then: 
\begin{enumerate}
\item there exist connected components $C_{\varepsilon, 1}, \hdots, C_{\varepsilon,l}$ of $V_\varepsilon \cap \R\langle \varepsilon\rangle^n$ such that 
$$C = \bigcup_{i=1}^l \lim_{\varepsilon \rightarrow 0} C_{\varepsilon, i}$$
\item if $C$ is bounded by some open ball $B \subset \R^n$ and does not intersect the boundary of $B$, then $C_{\varepsilon, i}$ is bounded by and does not intersect the boundary of $\text{Ext}(B , \R\langle \varepsilon \rangle^n)$ for $1 \leq i \leq l$. 
\end{enumerate}
\end{proposition}

\begin{proof} 
Let $\bz \in C$. Then there exists a connected component $S \subset \R^n \setminus \V(F)$  such that $\bz \in \overline{S}$. Then there exists some $\bp \in S$ such that $\bp \in B(\bz,r)$ for $r >0.$ We note that since $F$ is nonnegative over $\R^n$, $F(\bp) >0$. By \cite[Theorem 3.19]{BasuPollackRoyBook}, there exists a continuous semi-algebraic function $\gamma: [0,1] \rightarrow \overline{S}$ with $\gamma(0) = \bz$, $\gamma(1)=\bp$ and $\gamma(t)\in S$ for all $t \in (0,1]$. Note that we have $(F \circ \gamma)(0)=F(\bz)=0$ and $(F\circ \gamma)(1)=F(\bp)>0$. 

Let $F' :=\text{Ext}(F)$ and $\gamma':= \text{Ext}(\gamma)$, where $\text{Ext}$ is the extension as in Definition \ref{def:extension}. By the infinitesimal property of $\varepsilon$, $\text{Ext}([0,1])$  includes all Puiseux series with constant term in $[0,1]$. The Intermediate Value Theorem  applied to $F' \circ \gamma'$ gives some $t_\varepsilon \in \text{Ext}([0,1], \R\langle \varepsilon \rangle)$ such that $(F' \circ \gamma')(t_\varepsilon) = \varepsilon$. Let $\bz_\varepsilon :=\gamma'(t_\varepsilon)$. Then $\lim_{\varepsilon \rightarrow 0}(\bz_\varepsilon) = \bz$. Let $C_{\bz_{\varepsilon}}$ be the connected component of $V_\varepsilon \cap \R\langle \varepsilon \rangle^n$ containing $\bz_{\varepsilon}$, and associate $\bz$ to that component. Since there are finitely many connected components of $V_\varepsilon \cap \R \langle \varepsilon \rangle^n$, as we run through all $\bz \in C$, a subset of these connected components are of the form $C_{\bz_{\varepsilon}}$ with $\lim_{\varepsilon \rightarrow 0}(\bz_\varepsilon) = \bz$ for some $\bz\in C$. We denote these components of $V_\varepsilon \cap \R \langle \varepsilon \rangle^n$ by $C_{\varepsilon, 1}, \hdots, C_{\varepsilon,l}$. Clearly $C \subset \cup_{i=1}^l \lim_{\varepsilon \rightarrow 0} C_{\varepsilon, i}$.

Now suppose $\bz' \in \lim_{\varepsilon \rightarrow 0} C_{\varepsilon, i}$ for some $1 \leq i \leq l$. Then there exists some $\bz'_\varepsilon \in C_{\varepsilon, i}$ such that $\lim_{\varepsilon \rightarrow 0} \bz'_{\varepsilon} = \bz'.$ We recall that $C_{\varepsilon, i}$ is associated to some $\bz \in C$, i.e. there exists some $\bz_{\varepsilon} \in C_{\varepsilon, i}$ such that $\lim_{\varepsilon \rightarrow 0} \bz_{\varepsilon} = \bz$. Since $C_{\varepsilon, i}$ is connected, there exists some continuous semi-algebraic function $\gamma:[0,1] \rightarrow C_{\varepsilon, i}$ such that $\gamma(0) = \bz'_{\varepsilon}$ and $\gamma(1) = \bz$ and $\Gamma := \gamma([0,1])$ is a connected semi-algebraic set. By the preservation of closed and boundedness under semi-algebraic mapping, $\Gamma$ is closed and bounded and by Lemma \ref{lemma:limsemi}, $\lim_{\varepsilon \rightarrow 0} \Gamma$ is connected. Furthermore, we note that $\lim_{\varepsilon \rightarrow 0} \Gamma \subset V \cap \R^n, \lim_{\varepsilon \rightarrow 0}
(\gamma(0)) = \bz'$, and $\lim_{\varepsilon \rightarrow 0}(\gamma(1)) = \bz.$ Hence $\bz' \in C$ and $\cup_{i=1}^l \lim_{\varepsilon \rightarrow 0} C_{\varepsilon, i} \subset C$, thus we have shown (1) as required. 

Now suppose $C$ is bounded by some ball $B \subset \R^n$ and does not intersect the boundary of $B$. Let $\bz_\varepsilon \in C_{\varepsilon,i}$ such that $\lim_{\varepsilon \rightarrow 0} \bz_\varepsilon \in C$. For sake of contradiction, suppose $\bz'_\varepsilon \in \overline{C_{\varepsilon, i} \setminus\text{Ext}(B )}   $. Since $C_{\varepsilon, i}$ is connected, there exists a continuous semi-algebraic function $\gamma:\text{Ext}[0,1] \rightarrow C_{\varepsilon, i}$ such that $\gamma(0) = \bz_\varepsilon, \gamma(1) = \bz'_\varepsilon$ and $\Gamma := \gamma(\text{Ext}[0,1])$ is a connected semi-algebraic set. The Intermediate Value Theorem  applied to the polynomial defining the boundary of $B$ composed by $\gamma$ gives some $t_\varepsilon \in \text{Ext}[0,1]$ such that $\gamma(t_\varepsilon)$ is in the boundary of $\text{Ext}(B)$. Then $\lim_{\varepsilon \rightarrow 0} \gamma(t_\varepsilon)$ is in the boundary of $B.$

By the preservation of closed and boundedness under semi-algebraic mappings, $\Gamma$ is closed and bounded. 
By Lemma \ref{lemma:limsemi}, $\lim_{\varepsilon \rightarrow 0} \Gamma$ is connected. Then $\lim_{\varepsilon \rightarrow 0} \Gamma \subset C$ and
\mbox{$\lim_{\varepsilon \rightarrow 0} \gamma(t_\varepsilon) \in C$}. However, this contradicts $C$ intersecting the boundary of $B$. Thus $C_{\varepsilon, i}$ is bounded by (and does not intersect the boundary of) $\text{Ext}(B,\R\langle \varepsilon \rangle^n)$. 
\end{proof}

A natural question arises: Why it is necessary for us to perturb the sum of squares of the polynomials, rather than each polynomial separately, when we are working in this context over the real numbers? The following example illustrates what can go wrong over the reals, as opposed to the complex numbers. 

\begin{example}\label{ex:realcounter}
This example illustrates that for $s>1$ with   $V_\varepsilon :=\V(f_1 - a_1\varepsilon, \ldots, f_s - a_s\varepsilon)\subset \C\langle \varepsilon \rangle^n$ and $V:=\lim_{\varepsilon\rightarrow 0}V_\varepsilon\subset \C^n$, we may have connected components $C \subset V\cap \R^n $ that can only be extended to the complex part of $V_\varepsilon$, i.e. $C\not \subset \lim_{\varepsilon\rightarrow 0}\left(V_\varepsilon\cap \R\langle \varepsilon\rangle^n \right)$.

Let $s=n=2$, $f_1=x_1^2+x_2^2-1,$ $f_2=-(x_1-2)^2-x_2^2+1$. Then with $V_\varepsilon :=\V(f_1 - \varepsilon,f_2-\varepsilon)\subset \C\langle \varepsilon \rangle^2$ we have 
$$
\left(\lim_{\varepsilon\rightarrow 0}V_\varepsilon\right)\cap \R^2=\{(1, 0)\} 
$$
is a single point, so this point is the only connected component. Note that  any points in $V_\varepsilon$ have to satisfy  $f_1=f_2$, so in particular they will correspond to points on the intersections of the graphs of $f_1$ and $f_2$.  The graphs of $f_1$ and $f_2$ are the surfaces $P_1:=\V(x_3-
x_1^2-x_2^2+1)$ and $P_2:=\V(x_3+
(x_1-2)^2+x_2^2-1))$, respectively, which are two 3-dimensional parabolas. 

Over the complex numbers, $P_1$ and $P_2$ intersect in two complex lines
$$
\V(x_2 \pm i(x_1 - 1), x_3 - 2x_1 + 2)\subset \C^3
$$
so whenever $x_3=\varepsilon$, i.e. $x_1=\frac{\varepsilon}{2}+1$, the points 
$
\left(\frac{\varepsilon}{2}+1, \pm i\frac{\varepsilon}{2}\right)
$
are in $ V_\varepsilon$ and 
$$\lim_{\varepsilon\rightarrow 0}\left(\frac{\varepsilon}{2}+1, \pm i\frac{\varepsilon}{2}\right) = (1, 0).
$$ 
This also shows that  $$\lim_{\varepsilon\rightarrow 0}\left(V_\varepsilon\cap \R\langle \varepsilon\rangle^n\right)=\emptyset.$$
In particular,  over the reals
$P_1$ is a convex parabola with vertex $(0,0,-1)$, $P_2$ is a concave parabola with vertex $(2,0,1)$, and they tangentially intersect at $(1,0,0)$. So for any real $\varepsilon\neq 0$  we have 
$$
P_1 \cap P_2 \cap \V(x_3-\varepsilon)\cap \R^3=\emptyset.
$$

Alternatively, for $F:=f_1^2+f_2^2$, if we define $V_\varepsilon:=\V(F-\varepsilon)$, then we proved above that  
$$
\left(\lim_{\varepsilon\rightarrow 0} V_\varepsilon\right) \cap \R^2 = \lim_{\varepsilon\rightarrow 0}\left( V_\varepsilon \cap \R\langle\varepsilon\rangle^2\right).
$$
\end{example}


\section{Perturbations and Limits Polar Varieties}

There are many different definitions of polar varieties; for a survey and comparison of different notions see \cite{HarrisThesis2021}. In this paper  we  reduce to the hypersurface case by taking the sum of squares of the given real polynomials. For this case,  without loss of generality (in particular, when we apply a change of variables later in our paper), we can choose the appropriate number of partials to obtain a simplified definition of polar varieties. 
In practice, other notions of polar varieties may work better. We chose this presentation for its simplified notation and presentation, following the approach of \cite{MohabElias2013}, for conciseness. 

\begin{definition}\label{def:polar}
 Let $F\in \C[\bx]$ be square-free and $V=\V(F)\subset\C^n$. Consider the projections $\pi_i(x_1, \ldots, x_n)= (x_1, \ldots, x_i)$ for $i=1, \ldots, n$. The {\em  polar variety associated to $\pi_i$} of $V$ is defined as 
$$
\crit(V, \pi_i):=\V\left(F, \frac{\partial F}{\partial x_{i+1}}, \ldots,  \frac{\partial F}{\partial x_{n}}\right)\subset \C^{n} \quad i=1, \ldots, n, 
$$
based on how the polynomials defining this algebraic set correspond to the notion of critical points of a map. 
\end{definition}

We note that a significant difference in how this definition is stated compared to other notions of polar varieties is that it does not exclude the singular locus of an algebraic set $V$ from the polar varieties associated to $V$. We will address the smoothness of $V$ going forward via a change of variables and perturbations, so in fact it is natural to make this simplified modification for our context.

We use the following notation to perform a change of variables. 
\begin{definition}\label{def:changeofvar}
Let $F\in \R[\bx]$, $V=\V(F)\subset \C^n$,
and $A\in {\rm GL}_n(\R)$. 
Then,  we denote $F^A(x):=F(Ax)$, i.e. $V^A:=\V(F^A)$ is the image of $V$ via the map $x \mapsto A^{-1}x$. 
\end{definition}

Next, we state some known results on polar varieties which will be used in the proofs of our algorithms. In particular, polar varieties provide a nice way for us to lower the complex dimension of an algebraic set without losing the real points in the set. 

\begin{theorem}\label{thm:polardim} \cite[Proposition 3]{Banketal1997}
Let $F \in \R[\bx]$ be non-constant, square-free, and define a smooth algebraic set $V := \V(F) \subset \C^n$. Then there exists a non-empty Zariski open set $\mathcal{A} \subset {\rm GL}_n(\C)$ such that for all $A \in {\rm GL}_n(\R) \cap \mathcal{A}$ and $1 \leq i \leq n, \crit(V^A, \pi_i)$ is either empty or equidimensional of complex dimension $i-1.$
\end{theorem}

We note that in the above reference, the proof of this theorem consists of characterizing the set of matrices for which the result does not hold and showing that those matrices make up a Zariski closed set ${\rm GL}_n(\C) \setminus \mathcal{A}$, i.e. the complement of $\mathcal{A}$. 

\begin{corollary}\label{cor:polardim}
Let $F$ and $V$ be as in Theorem \ref{thm:polardim}. Suppose $\varepsilon$ is an infinitesimal and $V_\varepsilon := \V(F-\varepsilon) \subset \C\langle \varepsilon \rangle ^n$ is a smooth algebraic set on the field of Puiseux series as in Definition \ref{def:Puiseux}. Then there exists a non-empty Zariski open set $\mathcal{A} \subset {\rm GL}_n(\C\langle \varepsilon \rangle)$ such that for all $A \in {\rm GL}_n(\R) \cap \mathcal{A}$ and $1 \leq i \leq n, \crit(V^A_\varepsilon, \pi_i)$ is either empty or equidimensional of complex dimension $i-1.$ 
\end{corollary}

The proof of \ref{cor:polardim} follows from the following lemma. 

\begin{lemma}\label{lem:noepsilon}
Let $\mathcal{A}:= {\rm GL}_n(\C\langle \varepsilon \rangle)\setminus \V(Q)$ be a non-empty Zariski open subset of ${\rm GL}_n(\C\langle \varepsilon \rangle)$ defined by some polynomial $Q\in\R\langle \varepsilon \rangle[a_{i,j}]_{i,j=1}^n$. Then $\lim_{\varepsilon\rightarrow 0}\mathcal{A}\cap {\rm GL}_n(\R)$ is also a non-empty Zariski open subset of ${\rm GL}_n(\R)$.
\end{lemma}

\begin{proof}
Since $\mathcal{A}$ is non-empty, $Q\neq 0$. $Q$ is a polynomial in the variables $\{a_{i,j}\}_{i,j=1}^n$ with coefficients in $\R\langle \varepsilon \rangle$. We can assume, without loss of generality, that these coefficients are polynomials in $\varepsilon^{\frac{1}{q}}$ for some $q\in \N$ (by multiplying with a possible common denominator of these coefficients). Also, we can assume that $Q$ has minimal degree in $\varepsilon^{\frac{1}{q}}$ among all such polynomials defining ${\rm GL}_n(\C\langle \varepsilon \rangle)\setminus\mathcal{A}$. Thus
$$Q = Q_0 + \varepsilon^{\frac{t}{q}} Q_1$$
where $t \in \Z^+, Q_0 \in \R[a_{i,j}]$, and $Q_1 \in \R[\varepsilon^{\frac{1}{q}}][a_{i,j}]$. If $Q_0=0$, then $Q_1$ has lower degree than $Q$ in $\varepsilon^{\frac{1}{q}}$ and still defines ${\rm GL}_n(\C\langle \varepsilon \rangle)\setminus\mathcal{A}$, a contradiction. Therefore, $Q_0\neq 0$ and 
$$
\V(Q_0)=\lim_{\varepsilon\rightarrow 0} \V(Q)\neq {\rm GL}_n(\C)
$$
so $\lim_{\varepsilon\rightarrow 0}\mathcal{A}\cap {\rm GL}_n(\R)$ is also a non-empty Zariski open subset of ${\rm GL}_n(\R)$.
\end{proof}

The statements of the next two  propositions follow the approach of the unpublished work \cite{MohabElias2013}, so we restate and prove them here.

\begin{proposition}\label{prop:NN}
Suppose $F\in \R[\bx]$ and $\V(F) \cap \R^n$ is bounded. There exists a non-empty Zariski open set $O\in {\rm GL_n(\C)}$  such that for $A \in O\cap GL_n(\R)$, if $V^A=\V(F^A)$ and  $V^A_{\varepsilon}:=\V(F^A-\varepsilon)\subset \C\langle\varepsilon\rangle^{n}$ for $\varepsilon$ infinitesimal, then  \\
\textbf{(i)} for all $1 \leq i \leq n$, $\crit( V^A_{\varepsilon}, \pi_i)$ is either empty or is smooth and equidimensional with complex dimension $i-1$; \\
\textbf{(ii)} for all $\bp\in V^A\cap \R^{n}$, $\pi_d^{-1}(\pi_d(\bp))\cap (V^A\cap\R^{n})$ is finite, where $d$ is greater than or equal to the local real dimension of $V^A$ at $\bp$. 

\end{proposition}

\begin{proof} 
(i)  First note that  $V_{\varepsilon}$ is smooth by Lemma \ref{lem:perturbsmooth}. By Corollary \ref{cor:polardim}, we obtain a non-empty Zariski open set $O_1 \in {\rm GL_n(\C\langle \varepsilon \rangle)}$ such that for $A \in O_1 \cap {\rm GL_n(\R)}$, and  $1 \leq i \leq n$  $\crit(V_\varepsilon^A, \pi_i)$ is either empty or equidimensional with complex dimension $i-1$. \\

\noindent (ii) Let $V:=V(F)$. Since $V \cap \R^n$ is semi-algebraic, we can consider it as a union of connected components $C_1, \hdots, C_l$ with corresponding real dimension $d_1, \hdots, d_l$, as in Definition \ref{def:realdimension}. Then the local real dimension of $V \cap \R^n$ at $\bp$ is given by $\max_{\bp \in \overline{C_i}} d_i$, as in Definition \ref{def:reallocaldimension}. 

Let $V_i$ represent the Zariski closure of each $C_i$ for $1 \leq i \leq l$. Then the corresponding complex dimensions of $V_1, \hdots, V_l$ are also $d_1, \hdots, d_l$. By a version of Noether's normalization lemma (see for example \cite{Logar}), there exists a non-empty Zariski open set $O_{2,i} \in {\rm GL_n(\C)}$ such that for $A \in O_{2,i} \cap {\rm GL_n(\R)}$ and $\bq \in \C^{d_i}$, $\pi_{d_i}^{-1}(\bq) \cap V_i^A$ is finite. Then for $\bq \in \R^{d_i}$, $\pi_{d_i}^{-1}(\bq) \cap C_i$ is finite. 

Let $\bp \in V^A \cap \R^n$ where $A \in O_2 = \cap_{i=1}^l O_{2,i}$. Suppose $d \geq \max_{\bp \in \overline{C_i}} d_i$. Then $\pi_d(\bp) \in \R^d.$ For any $d_i =d$, taking $O_2$ as defined above guarantees $\pi_d^{-1}(\pi_d(\bp))\cap (V^A\cap\R^{n})$ is finite. Furthermore, for any $d_i$ strictly less than $d$, $\pi_d^{-1}(\pi_d(\bp))\cap (V^A\cap\R^{n})$ is still finite for $A \in O_2$ because $\pi_d^{-1}(\pi_d(\bp))\subset \pi_{d_i}^{-1}(\pi_{d_i}(\bp))$.  Taking $O = O_1 \cap O_2$ completes the proof. 
\end{proof}

Now we are ready to state the main result of this section, again emphasizing that it was originally proven in the unpublished paper \cite{MohabElias2013}.

\begin{proposition} \label{prop:limpolar}
Let $F, g_1, \ldots, g_m\in \R[\bx]$ and let $\varepsilon$  be infinitesimal. Suppose $F \geq 0$ on $\R^{n}$, $\V(F)\cap \R^{n}$ is bounded, and suppose that \textbf{(i)} and \textbf{(ii)} from Proposition \ref{prop:NN} hold with $A= I \in GL_n(\R)$ for $V=\V(F)$ and  $V_{\varepsilon}:=\V(F-\varepsilon)\subset \C\langle\varepsilon\rangle^{n}$. Then for $i=0, \ldots, n$, $\lim_{\varepsilon\rightarrow 0}\crit(V_{\varepsilon}, \pi_i)$ is equidimensional of dimension $i-1$ and for
$$U:=\{\bx\in \R^{n}\;:\; g_1(\bx)>0, \ldots, g_m(\bx)>0\}\subset \R^{n}
$$ and $S:=\V(F)\cap U$  we have 
$$\left(\lim_{\varepsilon\rightarrow 0}\crit(V_{\varepsilon}, \pi_i)\right)\cap U=S \quad \Leftrightarrow \quad \dim_\R(S)\leq i-1.
$$
\end{proposition}

\begin{proof} 
($\Rightarrow$) Fix some $0 \leq i \leq n$ and suppose $$\left(\lim_{\varepsilon\rightarrow 0}\crit(V_{\varepsilon}, \pi_i)\right)\cap U=S.$$ 
If $S = \emptyset,$ then $\dim_\R(S) = -1 \leq i-1$ and we are done. 
So assume $S \not = \emptyset.$ Since $S := \V(F) \cap U$, $\V(F) \cap \R^n \not = \emptyset$ and $S$ is bounded by our initial assumption of $\V(F) \cap \R^n$ being bounded.

By Lemma \ref{polarintersect}, for $1\leq i \leq n$, $\crit(V_\varepsilon, \pi_i) \not = \emptyset$. Then by Proposition \ref{prop:NN}, $\crit(V_\varepsilon, \pi_i)$ is equidimensional of complex dimension $i-1$. Then $\lim_{\varepsilon \rightarrow 0}(\crit(V_\varepsilon, \pi_i))$ has complex dimension $i-1$. So the real dimension of $\lim_{\varepsilon \rightarrow 0}(\crit(V_\varepsilon, \pi_i))$ is $\leq i -1$ and thus 
$$\left(\lim_{\varepsilon\rightarrow 0}\crit(V_{\varepsilon}, \pi_i)\right)\cap U=S$$
has real dimension $\leq i-1$. 

\noindent($\Leftarrow$) Now suppose $\dim_\R(S) \leq i-1$. Since $\crit(V_\varepsilon, \pi_i) \subset V_\varepsilon, \lim_{\varepsilon \rightarrow 0} V_\varepsilon = \V(F)$, and $\V(F) \cap U = S$, 
$$\left(\lim_{\varepsilon\rightarrow 0}\crit(V_{\varepsilon}, \pi_i)\right)\cap U \subset S.$$

If $S = \emptyset,$
$$S \subset \left(\lim_{\varepsilon\rightarrow 0}\crit(V_{\varepsilon}, \pi_i)\right)\cap U$$
and we are done. So suppose $S \not = \emptyset$ and take $\bz = (z_1, \hdots, z_n) \in S.$ Since $S = \V(F) \cap U$, $\bz \in U$ and $\bz \in \V(F) \cap \R^n$. Furthermore, the local real dimension of $S$ at $\bz$ is $\leq i-1$, so the local real dimension of $\V(F) \cap \R^n$ at $\bz$ is also $\leq i-1$.

 Define $F'$ as the function $F$ where the first $i-1$ coordinates have been evaluated at the first $i-1$ coordinate values of $\bz$,i.e.
$$F' := F(x_i, \hdots, x_n) = F(z_1, \hdots, z_{i-1}, x_i, \hdots, x_n).$$
We note that since $F$ is nonnegative over $\R^n$, $F'$ is nonnegative over $\R^{n-i+1}$.
Also define $\bz' :=(z_i, \hdots,z_n)$, $V':=\V(F')$ and $V'_\varepsilon := \V(F'-\varepsilon) \subset \C\langle \varepsilon \rangle^{n-i+1}$, and the canonical projection $\varphi_i(\bx) = x_i$ and the respective $\varphi_i'(x_i, \hdots, x_n) = x_i$. Note that $\bz'$ is  isolated in $V'\cap \R^{n-i+1}$ since $\pi^{-1}_{i-1}(\pi_{i-1}(\bz)) \cap \V(F) \cap \R^n$ is finite by \textbf{(ii)} of Proposition \ref{prop:NN}. 

Applying Lemma \ref{lem:limpolarone} to $\bz'$ and $V'$ we get some $\bz'_\varepsilon \in \crit(V'_\varepsilon,\varphi'_i)$ such that $\lim_{\varepsilon \rightarrow 0}\bz'_\varepsilon = z'=(z_i, \hdots, z_n).$ Define $\bz_\varepsilon := (z_1, \hdots, z_{i-1}, \bz'_\varepsilon)$ and $V^*_\varepsilon = V_\varepsilon \cap \text{Ext}(\pi_{i-1}^{-1}(z_1, \hdots, z_{i-1}))$, where $\text{Ext}$ is the extension from Definition \ref{def:extension}. Then $\bz_\varepsilon \in \crit(V^*_\varepsilon, \varphi_i).$ By Lemma \ref{lem:polarcontain}, $\bz_\varepsilon \in \crit(V_\varepsilon, \pi_i)$. Since $\lim_{\varepsilon \rightarrow 0}\bz_\varepsilon = \bz$,
$$S \subset \left(\lim_{\varepsilon\rightarrow 0}\crit(V_{\varepsilon}, \pi_i)\right)\cap U.$$
\end{proof}

The following lemmas were used in the proof of the above proposition.

\begin{lemma}\label{lem:limpolarone}
Assume $f_1, \ldots, f_s\in \R[\bx], F = f_1^2 + \cdots f_s^2$, and $\varepsilon >0$ a real infinitesimal. Let  $V_\varepsilon :=\V(F-\varepsilon)\subset \C\langle \varepsilon \rangle^n$ and  $V:=\lim_{\varepsilon\rightarrow 0}V_\varepsilon \subset \C^n$. Suppose $\bz \in V \cap \R^n$ and there exists a neighborhood $B(\bz, r) \subset \R^n$ for some $r>0$ such that $B(\bz,r) \cap V \cap \R^n$ is a finite set. Then there exists $\bz_\varepsilon \in \crit(V_\varepsilon, \pi_1)$ such that $\lim_{\varepsilon \rightarrow 0}\bz_\varepsilon = \bz.$ 
\end{lemma}

\begin{proof}
Since $B(\bz,r) \cap V \cap \R^n$ is a finite set, there exists some $r'>0$ such that $\bz$ is the only point in $B(\bz,r') \cap V \cap \R^n$. So $\{\bz\}$ is a bounded connected component of $V \cap \R^n$. Then by Proposition \ref{prop:connectedlimits}, there exist connected components $ C_{\varepsilon, 1}, \hdots, C_{\varepsilon, l} $ of $V_\varepsilon \cap \R\langle \varepsilon \rangle^n$ such that $\{\bz\} = \cup_{i=1}^l \lim_{\varepsilon \rightarrow 0} C_{\varepsilon, i}$ and $C_{\varepsilon, i}$ is bounded by and does not intersect the boundary of $\text{Ext}(B(\bz, r'))\subset  \R \langle \varepsilon \rangle^n$ for $1 \leq i \leq l$. 

Then for $1 \leq i \leq l$, we have $C_{\varepsilon, i} \subset \text{Ext}(B(\bz, r') )$ and is closed and bounded. By the extreme value theorem, $C_{\varepsilon, i} \cap \crit(V_\varepsilon, \pi_1) \not = \emptyset$. Since $\lim_{\varepsilon \rightarrow 0} C_{\varepsilon, i} = \{\bz\}$, all $\bz_\varepsilon \in C_{\varepsilon, i} \cap \crit(V_\varepsilon, \pi_1) \not = \emptyset$ are such that $\lim_{\varepsilon 
\rightarrow 0} \bz_\varepsilon = \bz$, and we are done. 
\end{proof}

\begin{lemma}\label{polarintersect}
Assume $f_1, \ldots, f_s\in \R[\bx], F = f_1^2 + \cdots f_s^2$, and $\varepsilon >0$ a real infinitesimal. Let  $V_\varepsilon :=\V(F-\varepsilon)\subset \C\langle \varepsilon \rangle^n$ and  $V:=\lim_{\varepsilon\rightarrow 0}V_\varepsilon \subset \C^n$ with $V\cap \R^n$ nonempty and bounded. Then $\crit(V_\varepsilon, \pi_i)$ is nonempty and intersects each bounded connected components of $V_\varepsilon \cap \R\langle \varepsilon \rangle^n$ for all $1 \leq i \leq n.$
\end{lemma}

\begin{proof}
Since $V \cap \R^n$ is nonempty, there exists some nonempty connected component $C \subset V \cap \R^n$. Let $\bz \in C$. Since $V \cap \R^n$ is bounded, there exists some $r>0$ such that the $C \subset B(\bz, r) \subset \R^n$ and $C$ does not intersect the boundary of $B(\bz,r)$. So by Proposition \ref{prop:connectedlimits}, there exist connected components $C_{\varepsilon, 1}, \hdots, C_{\varepsilon, l}$ of $V_\varepsilon \cap \R\langle \varepsilon \rangle^n$ such that $C = \cup_{j=1}^l \lim_{\varepsilon \rightarrow 0} C_{\varepsilon, j}$ and $C_{\varepsilon, j}$ is bounded by and does not intersect the boundary of $\text{Ext}(B(\bz, r)) \subset  \R \langle \varepsilon \rangle^n$ for $1 \leq j \leq l$. 
Thus, for $1 \leq j \leq l$, we have $C_{\varepsilon, j} $  is closed and bounded. Hence, by Lemma \ref{lem:limpolarone}, $C_{\varepsilon, j} \cap \crit(V_\varepsilon, \pi_1) \not = \emptyset$. Since by definition $\crit(V_\varepsilon, \pi_1) \subset \crit(V_\varepsilon, \pi_i),$ we are done.
\end{proof}

\begin{lemma}\label{lem:polarcontain}
Let $F \in \C[\bx]$ and $\alpha = (\alpha_1, \hdots, \alpha_{i-1}) \in \C^i.$ Suppose $V_{i,\alpha}$ is the algebraic set $\V(F) \cap \pi_{i-1}^{-1}(\alpha)$ and $\varphi_i$ is the projection defined by $\varphi_i(\bx) = x_i$. Then 
$$\crit(V_{i, \alpha},\varphi_i) \subset \crit(\V(F),\pi_i).$$
\end{lemma}

\begin{proof}
We recall that by definition
$$\crit(\V(F), \pi_i) = \V\bigg(F, \frac{\partial F}{\partial x_{i+1}}, \hdots, \frac{\partial F}{\partial x_n}\bigg).$$
By how we have defined $V_{i,\alpha}$, $\crit(V_{i, \alpha},\varphi_i)$ is the algebraic set defined by the polynomials $F, x_1 -\alpha_1, \hdots, x_{i-1} - \alpha_{i-1}$ and the maximal minors of the Jacobian matrix of the polynomials. Then in fact, 
$$\crit(V_{i, \alpha},\varphi_i) = \V\bigg(F, x_1 -\alpha_1, \hdots, x_{i-1} - \alpha_{i-1}, \frac{\partial F}{\partial x_{i+1}}, \hdots, \frac{\partial F}{\partial x_n}\bigg)$$
and we are done. 
\end{proof}

\section{Shifting from Infinitesimals to Complex Perturbations} \label{sec:complpert}

In this section, we establish results in order to formulate our algorithms so they can be implemented not only purely symbolically, but also in a numerical algebraic geometry context. Here we track our perturbed set to its limit variety by employing homotopy continuation while our perturbation constant follows a complex arc towards zero. 

To this end, we first shift from the paradigm of real infinitesimals to arbitrarily small real numbers, as established by the following result from real algebraic geometry. 

\begin{theorem}\label{thm:shifttoreal} \cite[Proposition 3.17]{BasuPollackRoyBook}
A result holds over $\R\langle \varepsilon \rangle$ if and only if there exists some $e_0 \in \R$ such that it also holds for all $e \in (0, e_0) \cap \R$. 
\end{theorem} 

For our purposes, we also want to establish that we are able to make this switch in terms of witness set computations, which are done over the complex numbers. Therefore, the results in this section can be formulated in terms of a more general perturbation setup, based on the following result from Faug\'ere et al. on perturbing the defining polynomials of an algebraic set. 

\begin{lemma}\label{lem:genericity} \cite[Lemma 1]{Faugereetal2008}
Let $f_1, \ldots, f_s\in \R[\bx]$ and fix $l\leq s$ and $\{i_1, \ldots, i_l\}\subset \{1, \ldots, s\}$. Then there exists a Zariski closed subset ${\mathcal A}\times {\mathcal E}\subset \C^s\times \C $ such that for all $\ab:= (a_1, \ldots, a_s)\in \R^s\setminus {\mathcal A}$ and $e\in \R\setminus {\mathcal E}$,
the ideal generated by the polynomials 
$
f_{i_1}-ea_{i_1}, \ldots, f_{i_l}-ea_{i_l}
$
is a radical equidimensional ideal and $\V(f_{i_1}-ea_{i_1}, \ldots, f_{i_l}-ea_{i_l})$ is either empty or smooth of dimension $n-l$. 
\end{lemma}

\begin{remark}
We will use these more general deformations in Sections \ref{sec:complpert}, \ref{sec:defwit} and \ref{sec:compg}, where the results are valid over $\C$. In the previous sections and in Section \ref{sec:smooth2}   we need to use a more specialized perturbation of the polar varieties for our results  to hold over the reals. See Example \ref{ex:realcounter} why the more general deformations may not work over the reals. 
\end{remark}

Using Lemma \ref{lem:genericity}, we define a genericity assumption which holds over the complex numbers. 

\begin{definition}\label{def:AA}
Consider polynomials $f_1, \ldots, f_s\in \R[\bx]$ and point
\mbox{$\ab=(a_1, \ldots, a_s)\in \R^{s}$}. We say that $f_1, \ldots f_s $ and $\ab$ satisfy Assumption {\tt (A)} if 
\begin{enumerate}
\item[{\tt (A)}:] There exists $e_0 > 0$ such that for all $0<e\leq e_0$, 
the polynomials $f_{1}-ea_{1}, \ldots, f_{s}-ea_{s}$
generate a radical equidimensional ideal and $V^\ab_e := \V(f_{1}-ea_{1}, \ldots, f_{s}-ea_{s})$ is  smooth and has dimension $n-s$.
\end{enumerate}
\end{definition}

Assumption ({\tt A}) in Definition \ref{def:AA} guarantees the existence of some $e_0>0$; however, in practice this number can be arbitrarily small. Instead of trying to compute an $e_0$ that works for a given system $f_1, \ldots, f_s$,  the next result shows that we can choose a generic $\xi\in \C$ with $|\xi|=1$ to replace $e_0$ with $\xi$ and $e$ with $t\xi$, where $t\in (0,1]$.

\begin{proposition}
 \label{lem:limpolarcomplex}
Let $f_1, f_2, \ldots, f_s\in \R[\bx]$, $\ab=(a_1, \ldots, a_s)\in \R^s$ and let $\varepsilon$  be infinitesimal.  Assume that  $V^\ab_\varepsilon:=\V(f_1-\varepsilon a_1, \ldots, f_{s}-\varepsilon a_s)\subset \C\langle\varepsilon\rangle^n $ is smooth and equidimensional of dimension $n-s$. Then for all but finitely many   $\xi\in \C$ with $|\xi|=1$ and for all $ t\in (0,1]$,   $V^\ab_{t\xi}:=\V(f_1-t\xi a_1,\ldots, f_{s}-t\xi a_s)\subset \C^{n}$  is smooth and equidimensional of dimension $n-s$ and  in that case we have $$
\lim_{\varepsilon\rightarrow 0}V^\ab_{\varepsilon}=\lim_{t\rightarrow 0}V^\ab_{t\xi}.
$$
\end{proposition}

\begin{proof}
First, we show that for all but a finite number of choices of $\xi\in \C$,  $V_\xi^\ab=\V(f_1-\xi a_1,\ldots, f_{s}-\xi a_s) $ is smooth by proving that the set of ``bad" choices $\xi$ is a proper Zariski closed subset of  $\C$, thus finite. Note that from our assumptions on $V^\ab_\varepsilon$ we get that $f_1, \ldots, f_s$ and $\ab$ satisfy Assumption {\tt (A)} for some $\e_0>0$. Consider the ideal using new variables $x_0$,  $z$ and $\lambda_1, \ldots, \lambda_s$:
\begin{eqnarray*}I:=\langle f_1^{(h)}-a_1z x_0^{\deg(f_1)},\ldots, f_{s}^{(h)}- a_s z x_0^{\deg(f_s)}\rangle \\
+\langle(\lambda_1\nabla(f_1)+\ldots+\lambda_s \nabla(f_s))^{(h)} \rangle.
\end{eqnarray*}
Here $g^{(h)}$ denotes the homogenization of $g\in \R[x_1, \ldots, x_n]$ by the variable $x_0$ and $\nabla$ is the differential operator in the variables $x_1, \ldots, x_n$.  Thus $I$ is bi-homogeneous in the variables $(\lambda_1,\ldots \lambda_s)$ and $(x_0, \ldots, x_n)$. Then the projection of $X(I)\subset \PP^{n}\times \PP^s\times \C$ onto $\C$ is a Zariski closed subset of~$\C$, and since $e_0$ is not in the projection, the projection is not $\C$, thus a finite set  $Z$. Clearly, for $\xi\in \C\setminus Z$  and for all $\bp\in V^\ab_{\xi}$, the Jacobian of $f_1-\xi a_1,\ldots, f_{s}-\xi a_s$ at $\bp$ has rank $s$, thus $V^\ab_{\xi}$ is smooth and equidimensional of dimension $n-s$.    This also implies that for all but finitely many $\xi\in \C$ with $|\xi|=1$ and for all $t\in (0,1]$ we have that  $V_{t\xi}^\ab=\V(f_1-t\xi a_1,\ldots, f_{s}-t\xi a_s) $ is smooth and equidimensional.

Fix  $\xi\in \C\setminus Z$ with $|\xi|=1$ so $V_{t\xi}^\ab$ is smooth and equidimensional. 
To prove the second claim, let $L_1,\hdots, L_{n-s} \in \C[\bx]$ be linear polynomials such that $\LL := \V(L_1, \hdots, L_{n-s})$ is a generic linear space of codimension $n-s$ which intersects both $\lim\limits_{\varepsilon\rightarrow 0}V^\ab_{\varepsilon}$ and $\lim\limits_{t\rightarrow 0}V^\ab_{t\xi}$ transversely. By our assumptions,  both $V^\ab_{\varepsilon}  \cap \LL\subset \C\langle \varepsilon\rangle ^n$ and $V^\ab_{t\xi} \cap \LL\subset \C^n$ are finite for any fixed $t\in (0,1]$.

Then since $\LL$ does not depend on either $\varepsilon$ or $t$, 
\begin{align*} 
\lim_{\varepsilon \rightarrow 0} \bigg(V^\ab_{\varepsilon} \cap \LL \bigg)  &= \lim_{\varepsilon \rightarrow 0} V^\ab_{\varepsilon} \cap \LL 
\;\;\text{ and }\;\; \\
\lim_{t \rightarrow 0} \bigg(V^\ab_{t\xi} \cap \LL \bigg)  &= \lim_{t\rightarrow 0} V^\ab_{t\xi} \cap \LL.
\end{align*} 
Since $\LL$ is a generic linear space which intersects both $\lim\limits_{\varepsilon\rightarrow 0}V^\ab_{\varepsilon}$ and $\lim\limits_{t\rightarrow 0}V_{t\xi}$ transversely, 
 $\lim_{\varepsilon \rightarrow 0} V^\ab_{\varepsilon} \cap \LL =  \lim_{t\rightarrow 0} V^\ab_{t\xi} \cap \LL $
implies
$ \lim_{\varepsilon \rightarrow 0} V^\ab_{\varepsilon} =  \lim_{t\rightarrow 0} V^\ab_{t\xi}.$

So it is sufficient to prove that  
$$\lim_{\varepsilon \rightarrow 0} \bigg(V^\ab_{\varepsilon} \cap \LL \bigg)  = \lim_{t \rightarrow 0} \bigg(V^\ab_{t\xi} \cap \LL \bigg)$$
to achieve the desired result.

Let $H \subset \R[\bx, \varepsilon] $ be the system
$$H:= H(\bx, \varepsilon) =
 \left[  
f_1-\varepsilon a_1, \ldots, f_{s}-\varepsilon a_s, L_1,
\ldots, L_{n-s} \right]. $$
Let $S \subset \C \langle \varepsilon \rangle^n$ be the finite set of bounded solutions of $H = 0$, where bounded is as defined for Puiseux series in Definition \ref{def:Puiseux}. Then for all $\bx(\varepsilon) \in S$, let  
$ \lim_{\varepsilon \rightarrow 0} \bx(\varepsilon) = \bx_0 \in \C^n$. Furthermore, by the definition of $H$, 
$\lim_{\varepsilon \rightarrow 0} S = \lim_{\varepsilon \rightarrow 0} \bigg( V^\ab_\varepsilon \cap \LL\bigg).$

Using Puiseux's Theorem (e.g.,  see \cite[Chap.~7]{Fischer2001}
and \cite[\S~10.2, Thm.~A.3.2, Cor.~A.3.3]{SommeseWampler2005}),
one can ensure that each of the finitely many Puiseux series under consideration
has a positive radius of convergence.  
In particular, since $\varepsilon > 0$ is a real infinitesimal and $x(\varepsilon)\in S$ is 
bounded, each $\bx(\varepsilon)\in S$ has an interval of 
convergence $(0, e_x) \subset \R$ for some $e_x > 0$. 
Choose $e_0 > 0$ such that $e_0 < \min\limits_{x \in S} e_x$. 
Now we make a switch, and instead of considering $\bx(\varepsilon)\in S$ an element $\C \langle \varepsilon \rangle^n$, we consider $\bx$ as a function $\C\rightarrow \C^n$ which is well-defined  for $z \in \C$ with $|z| \leq e_0$. Abusing the notation, we denote by $\bx$ both the Puiseux series and  the corresponding complex function.

Recall that if a pair $( \bx^*, z^*)\in \C^n\times \C$ has the property that $H(\bx^*,z^*) = 0$ and $\det JH(\bx^*,z^*) = 0$, where $JH$ is the Jacobian matrix of $H$ with respect to the $\bx$ variables, then $z^*$ is a critical point and $\bx^*$ is a branch point for  for $H(\bx,z) = 0$. Let $\Cc$ denote the set of all critical points of $H(\bx,z) = 0$.  
Then, since $|S| < \infty$, we know $|\Cc| < \infty$.

Now let $z \in \Cc$. Then there exists some $\xi_z\in \C$ with $|\xi_z|=1$ such that for $t\in \R$, the path $\xi_zt$ passes through~$z$, so that $\bx(t\xi_z) \in \C^n$ has some branching point. Let 
$Z = \{  \xi_z :  z \in \Cc\}$, a subset of the unit circle in $\C$.  
Since $|\Cc| < \infty$, $|Z| <\infty$. Then, for any $\xi \in \C \setminus Z$ with $|\xi|=1$, we have that $\bx(t\xi) \in \C^n$ for $t \in (0,1]$ does not pass through branching points. Since $\C\setminus Z$ is Zariski dense in $\C$, the same
 holds for generic $\xi \in \C$ with $|\xi|=1$.

So let $\xi \in \C\setminus Z$ with $|\xi|=1$  and $H_{\xi} \subset \C^{n+1}$ be the homotopy defined by the system
\begin{equation*}
H_{\xi}:=H_{\xi}(\bx,t) = 
\left[  
f_1-t\xi a_1, \ldots, f_{s}-t\xi a_s,L_1 \ldots, L_{n-s} 
\right].
\end{equation*}
The limit points of the solutions of $H_{\xi}$ are $\lim\limits_{t \rightarrow 0} \bigg( V^\ab_{t\xi} \cap \LL \bigg)$.
Let $T \subset \C^n$ be the roots of $H_{\xi}(\bx,1)$. Then $|T| = | V^\ab_{\varepsilon} \cap \LL| < \infty$. Furthermore, by the above argument, the homotopy paths for $H_{\xi}$ are exactly described by the points in $V^\ab_{\varepsilon} \cap \LL \subset \C \langle \varepsilon \rangle^n$ by replacing~$\varepsilon$ with $t\xi$. Hence,
$$\lim_{\varepsilon \rightarrow 0} \bigg(V^\ab_{\varepsilon} \cap \LL \bigg)  = \lim_{t \rightarrow 0} \bigg(V^\ab_{t\xi} \cap \LL \bigg).$$
\end{proof} 

The following illustrates why we take $\xi \in \C \setminus \R$, a process generally known as 
the ``gamma trick,'' e.g., see \cite[Chap.~7]{SommeseWampler2005}

\begin{example}
Let $f(x) = x^3 - 3x^2 + 2x$, $a=1$, and $H_\xi(x,t) = f(x)-t\xi$.
For $\xi=1$, since $H_\xi(x,0) = x(x-1)(x-2)=0$ has three real solutions,
$H_\xi(x,1) = f(x)-1=0$ has one real solution, 
and $H_\xi(x,t)$ has real coefficients, there must be a value of $t\in(0,1)$
such that $H_1(x,t)=0$ has a singular solution, namely $t\approx 0.3849$.  
However, for, say, $\xi = 1 + \sqrt{-1}$, then $H_\xi(x,t)=0$ has three
nonsingular solutions for all $t\in[0,1]$.
A similar statement holds for $\xi \in \C \setminus \R$ 
with algebraic probability one. 
\end{example}

Proposition \ref{lem:limpolarcomplex} gives a proof of correctness for {\sc Witness Points in Limits Algorithm \ref{Alg:WPL}}
which computes a {\em witness point set} (as in Definition \ref{def:witnessset}) of a limit
with algebraic probability one. In Step (iv) of Algorithm \ref{Alg:WPL} we use numerical homotopy continuation method, see \cite{BertiniBook2013} for more details.

\begin{algorithm}[ht]
\caption[Alg:WPL]{\sc WitnessPointsInLimits }
\label{Alg:WPL}
\begin{description}
\item[Input:]  $f_1, \ldots, f_s\in \R[x_1, \ldots, x_n]$, $\ab=(a_1, \ldots, a_s)\in \R^s$ 
and $L=\{L_1, \ldots, L_{n-s}\}\subset \R[x_1, \ldots, x_n]$ generic linear 
polynomials. 
\item[Output:] {${\rm flag}=$TRUE if $\V(f_1-a_1e, \ldots, f_s-a_s e) \cap V(L)$ is 0-dimensional for all sufficiently small $e>0$ and  the finite set of points  in $W:=\lim_{e\rightarrow 0+}\left( \V(f_1-a_1e, \ldots, f_s-a_s e)\cap V(L)\right)$, ${\rm flag}=$FALSE otherwise and $W = \emptyset$.  }
\end{description}
\begin{enumerate}
\item Loop 
\begin{enumerate}[(i)]
\item Choose generic $\xi\in \C$ with $|\xi|=1$. 

\item Define $H_{\xi}(\bx,t) := 
\left[  
f_1-t\xi a_1, \ldots, f_{s}-t\xi a_s, L_1, \ldots, L_{n-s}  
\right].
$
\item If $|\V(H_{\xi}(\bx,1))| = \infty$, exit loop and return ${\rm flag} =$ FALSE, $W = \emptyset$. 

\item Compute $\lim_{t \rightarrow 0} \V(H_{\xi}(\bx,t))$ via a homotopy continuation, starting at $t=1$.

\item If no branch points were hit during homotopy tracking, exit loop and return ${\rm flag}=$ TRUE, $W = \lim_{t \rightarrow 0} \V(H_{\xi}(\bx,t)).$

\end{enumerate}
\end{enumerate} 

\end{algorithm}

\section{Deflated Witness Systems for Limits}\label{sec:defwit}

In this section, we apply some of the main ideas from numerical algebraic geometry following \cite{BertiniBook2013} to complex algebraic sets obtained as limits of perturbed positive dimensional algebraic set. 
We will use the notion of a {\em deflated witness~system}:

\begin{definition}
Let $V\subset \C^n$ be an equidimensional algebraic set and $(F,L,W)$ a witness set for $V$ as in Definition \ref{def:witnessset}. Then if each irreducible component of $V$ has multiplicity one with respect to $F$, $F$ is called a {\em deflated witness system} and $(F,L,W)$ is 
a {\em deflated witness~set} for $V$. 
\end{definition}

When trying to compute a deflated witness system for a variety defined by a limit, difficulties that arise are that the limit points may be singular, arising from multiple paths converging to the same limit point,   
or that the witness system $f=(f_1, \hdots, f_s)$ for the original algebraic set $\V(f)$ is not a witness system for $\lim_{t \rightarrow 0} V^{\ab}_{t\xi}=\lim_{t \rightarrow 0} \V(f_1-a_1t\xi, \ldots, f_s-a_s t\xi)$. This is demonstrated in the following example. 

\begin{example}\label{ex:subtle}
Consider $f(x,y) = (xy, xy-x) \subset \C[x,y]$ and $\ab =(1,0)$. Then for $\xi=1$ we get $V^{\ab}_{t\xi}=V(xy-t, xy-x) $, so for all fixed $t\in (0,1]$ we get $V^{\ab}_{t\xi}=\{(t,1)\}$, so  $\lim_{t\rightarrow 0} V^{\ab}_{t\xi} = \{(0,1)\}$. 
But $ V(f(x,y)) = \V(x)$ is not a witness system for $\{(0,1)\}$. We show that applying a straightforward  isosingular deflation to $f$ does not provide a deflated witness system for $\lim_{t\rightarrow 0} V^{\ab}_{t\xi} = \{(0,1)\}$. We compute 
\[ Jf(x,y) = \begin{bmatrix} 
		y & x \\
		y-1 & x \\ 
		\end{bmatrix}. \] 
For $\bp = (0,1)$, $\rank(Jf(0,1)) = 1$, so the determinant of $Jf(0,1)$. Thus isosingular deflation gives  $F = [xy, xy-x, x]$. 
But $\V(F) = \V(x)$, so $F$ is not a witness system for $\{(0,1)\}$. 

\end{example}

To overcome  these difficulties, we will use Theorem \ref{thm:isointersect} to compute a deflated witness system for the intersection of $V(f_1-a_1t\xi, \ldots, f_s-a_s t\xi)\cap V(t) \subset \C^{n+1}$ (considering $t$ as an extra complex variable, but  $\xi\in \C$ is fixed). 
 The {\sc Deflated Witness System Algorithm \ref{Alg:DWS}} computes a deflated witness system for irreducible components of a variety defined as a limit.

\begin{algorithm}[ht]
\caption[Alg:DWS]{\sc DeflatedWitnessSystem }
\label{Alg:DWS}
\begin{description}
\item[Input:] $f_1, \ldots, f_s\in \R[\bx]$, $\ab=(a_1, \ldots, a_s)\in \R^s$, and $\bp \in V:=\lim_{e\rightarrow 0+}\V(f_1-a_1e, \ldots, f_s-a_s e )$, a generic point on a unique irreducible component $V_\bp$ of $V$.

\item[Output:] A deflated witness system $G\subset \R[\bx]$ for $V_\bp$.
\end{description}

\begin{enumerate}
\item Define $F_0(\bx,t):=(f_1-a_1t, \ldots, f_s-a_s t)\in \R[\bx, t]^s$
and $\bq:=(\bp,0)\in \R^{n+1}$.
\item $F:=\textsf{IsosingularDeflation}(F_0,\bq)$. 
  {\small {\tt // See \cite[Alg.~6.3]{HauensteinWampler2013}}}

\item Define $G_0(\bx):=F(\bx,0)$.

\item $G:=\textsf{IsosingularDeflation}(G_0,\bp).$ \!{\small {\tt // See \cite[Alg.~6.3]{HauensteinWampler2013}}}

\item Return $G$.
\end{enumerate} 
\end{algorithm}

\begin{example}[Example \ref{ex:subtle} cont]
We apply Algorithm \ref{Alg:DWS} for this example. 
We have to do isosingular deflation of $F_0(x,y,t):=(xy-t, xy-x)$, as in Steps (1) and (2). Here, we compute
\[ JF_0(x,y,t) = \begin{bmatrix} 
		y & x & -1 \\
		y-1 & x & 0\\ 
		\end{bmatrix}. \] 

For $\bq = (0,1,0)$, $\rank(JF_0(0,1,0)) = 1$. We get in Step (2)  a deflated  system   $F(x,y,t) = [xy, xy-x,x, y-1]$. Since $G_0(x,y):=F$  is already deflated  for $\bp=(0,1)$, $G=G_0$ is the output.
\end{example}

\begin{theorem}\label{thm:deflcorrect}
Let $f_1, \ldots, f_s$, $\ab$, and $\bp$ as in the input of {\sc Algorithm \ref{Alg:DWS}}. Then $G$, computed by {\sc Algorithm~\ref{Alg:DWS}}, satisfies the output specifications.
\end{theorem}

\begin{proof}
Since $V_\bp$ is an irreducible component of $V$, there exists an irreducible component $Z\subset \V(F_0(\bx,t))\subset \C^{n+1}$ such that   $V_\bp\times\{0\}$ is an irreducible component of $Z\cap \V(t)$
which is an intersection.  Hence, one can apply the isosingular deflation approach applied
to intersections in Theorem \ref{thm:isointersect}.
Although Theorem \ref{thm:isointersect}
would deflate $H_0(\bx,t,t'):=(F_0(\bx,t), t')$ at $\bq':=(\bp, 0, 0)$,
the simplicity of the intersection together with $t'$ contained in $H_0$ easily
shows that one obtains an equivalent deflation as deflating $F_0(\bx,t)$ at $\bq=(\bp,0)$,
resulting in the system  $F(\bx,t)$.
Therefore, $V_\bp$ must be an irreducible component of $\V(F(\bx,0))$ 
so $G_0(\bx):=F(\bx, 0)$ is a witness system for $V_\bp$.  
Since $G_0$ need not be a deflated witness system for $V_\bp$,
one deflates $G_0$ at $\bp$ to yield a deflated witness system $G$ for $V_\bp$.
\end{proof}

\section{Computation of g} \label{sec:compg}

The final key tool required to compute a real smooth point on every bounded connected component of an algebraic set $V$ is a ``well-chosen" polynomial
$g$ that satisfies the conditions of Theorem~\ref{thm:noperturb}, 
i.e.,  $\Sing(V)\cap \R^n\subset \V(g)$ and  $\dim ( V\cap \V(g))< \dim (V)$. 
There exist symbolic methods  to compute such~a~$g$ for an equidimensional variety  $V$ with $\dim(V)=n-s$. 
For example, \cite[Lemma~4.3]{SafYanZhi2018} computes the defining equation $\omega\in \R[x_1, \ldots, x_{n-s+1}]$ of the projection~$\overline{\pi_{n-s+1}(V^A)}$ of a generic linear transformation $V^A$ of $V$ such that this projection  is a hypersurface.  Then, $g$ can be taken
to be one of the partial derivatives of $\omega$.
This idea could be extended to the case when $V$ is not equidimensional using infinitesimal deformations and limits (c.f., \cite{SafEl2018}). In our
{\sc Computation of g Algorithm \ref{Alg:Cg}}, we provide a new approach based on isosingular deflation, as discussed in Subsection~\ref{sec:Iso}, which computes several $g$'s depending on the isosingular deflation sequence of the irreducible components.

\begin{algorithm}[H]
\caption[Alg:Cg]{\sc ComputationOfG }
\label{Alg:Cg}
\begin{description} 
\item[Input:]  $f_1, \ldots, f_s\in \R[\bx]$, $\ab=(a_1, \ldots, a_s)\in \R^s$. 

\item[Output:] $\left\{\left(g_j, (G_j, L, W_j)\right)\;:\;j=1, \ldots, r\right\}$ such that for all $i\neq j\in \{1, \ldots, r\}$, $
V_e^\ab:=\V(f_1-a_1e, \ldots, f_s-a_s e)$ and  $V:=\lim_{e\rightarrow 0+}V_e^\ab $:   
\begin{enumerate}[(i)]
\item $g_j\in \R[\bx]$, $G_j, L\subset \R[\bx]$, and $W_j\subset V$. 
 \item $(G_j, L, W_j)$ is a deflated witness set of some $V_j\subset V$, where $V_j$ is a union of irreducible components of $V$;
\item $V=\bigcup _{j=1}^r V_j$
\item $\Sing(V_j)\subseteq \V(g_j)$
\item $\dim (V_j\cap \V(g_j))< n-s$
\item $\dim (V_i\cap V_j)< n-s$ and $V_i\cap V_j\subseteq \V(g_j). $
\end{enumerate} 
\end{description}
\begin{enumerate}
\item Loop 

\begin{enumerate}
\item Choose a generic system $L\subset \R[\bx]$ of $n-s$ linear polynomials.


\item $({\rm flag},W):=$ \textsf{WitnessPointsInLimits}$(\{f_1, \hdots, f_s\},\textbf{a}, L).$ 
\qquad {\small {\tt // See Algorithm \ref{Alg:WPL}}}

\item If ${\rm flag}=$TRUE, exit loop. 

\end{enumerate}

\item Set $j:=1$. 
\item Loop
\begin{enumerate}
\item Pick some $\bp\in W$.

\item $W_j:=\{\bp\}$.

\item Update $W:=W\setminus\{\bp\}$.

\item $G_j:= \textsf{DeflatedWitnessSystem}(\{f_1, \ldots, f_s\}, \ab, \bp)$. 
\qquad {\small {\tt // See Algorithm \ref{Alg:DWS}}}

{\small {\tt // $G_j \subset \R[\bx]$  is a witness system for the irreducible component $V_\bp\subset V$ containing $\bp$ such that
  $f_1, \hdots, f_s \in G_j$,  $G_j(\bp)=0$ and $\rank \,JG_j(\bp)=s$.}}

\item For all $\bp'\in W$ 

\qquad  If $G_j(\bp')=0$ and $\rank \,JG_j(\bp')=s$, then 

\qquad  \qquad    Update $W_j := W_j\cup\{\bp'\}$ and  $W := W\setminus\{\bp'\}$.

\item Compute $g_j(\bx):=\det(M(\bx))$, where $M$ is a generic rational linear combination of all $s \times s$ 
submatrices of $JG_j(\bx).$

\item If $W\neq \emptyset$, increment $j:=j+1$. 
\end{enumerate}
\end{enumerate}
\end{algorithm}

\begin{theorem}\label{thm:gcorrect}
Let $f_1, \ldots, f_s$, $\,\ab$,  $V_e^\ab$, and $V$  be as in the input and output specifications of 
{\sc Algorithm \ref{Alg:Cg}}. Assume that $
V_e^\ab:=\V(f_1-a_1e, \ldots, f_s-a_s e)$ satisfies Assumption {\tt (A)}.  Then {\sc Algorithm \ref{Alg:Cg}} is correct.
\end{theorem}

\begin{proof}
By our assumption on the genericity of $L$ and Assumption {\tt (A)}, $W$ is finite and each point $\bp \in W$ is a generic point of a unique irreducible components $V_\bp$ of $V$ containing $\bp$. Based on the output of Algorithm \ref{Alg:DWS}, assume that for any $\bp\in W$, in Step (3d) we  compute  $G_j\subset \R[\bx]$ such that the irreducible component $V_{\bp}\subset V$ containing $\bp$ is an irreducible component of $\V(G_j)$, $f_1, \ldots, f_s\in G_j$,  $G_j(\bp)=0$ and $\rank \,JG_j(\bp)=s$. 
Then, $G_j \subset \R[\bx]$ computed in Step (3d) deflates all generic points of $V_\bp$. 
Step (4) adds all other points from $W$ which are deflated by $G_j$.  In particular,
every other point on $V_\bp$ contained in $W$ will be added to $W_j$.
Hence, $(G_j, L , W_j)$ is a deflated witness set for a union of irreducible components of $V$, denoted by $V_j$, proving (ii). Since $\bigcup_j W_j=W$, 
we also get $\bigcup_j V_j=V$, which proves (iii).  
If $\by\in \Sing (V_j)$, then $\rank (JG_j(\by))<s$ so all $s\times s$ minors of $JG_j(\by)$ vanish.
Hence, $g_j(\by)=\det(M(\by))=0$ proving (iv).  
Conversely, for any $\bp'\in W_j$, some $s\times s$ minor of $JG_j(\bp')$ does not vanish at $\bp'$.  Since $g_j$ is a generic choice of combinations of all such minors, $g_j(\bp')\neq 0$ for all  $\bp'\in W_j$. By Assumption {\tt (A)}, $V=\lim_{e\rightarrow 0} V_e^\ab$ is equidimensional of dimension $n-s$, so for all $\bp'\in W$, $\dim V_{\bp'}=n-s$. Since $g_j$ does not vanish identically on $V_{\bp'}$ for any $\bp'\in W_j$, we get $\dim (V_j\cap \V(g_j))< n-s$, proving (v). 

To prove the first claim in (vi), note that each $V_i$ is a union of $(n-s)$-dimensional irreducible components of $V$ and sample points from the irreducible components of $V$ are uniquely assigned to one $W_j$.  Then for $i\neq j$, $V_i$ and $V_j$ cannot share an irreducible component, so their intersection is lower~dimensional. 

To prove the second claim in (vi) we use Theorem \ref{thm:singdeflation} as follows.  
Let $\by\in V_i\cap V_j$.  Suppose that $X$ is an irreducible component of $V_i$ and $Y$ is an irreducible component of $V_j$ such that $\by\in X\cap Y$.  
Let $\xi\in \C$ be generic with $|\xi|=1$, $t$ a complex variable, and denote $f_{\xi}^\ab=f_{\xi}^\ab(x,t) :=(f_1-a_1 t \xi, \ldots, f_s-a_s t \xi)$. 
Then, $X\times \{0\}$ and $Y\times \{0\}$ are irreducible varieties of $\C^{n+1}$ and both are subsets of $V(f_{\xi}^\ab)\subset \C^{n+1}$. 
Therefore, each is contained in a unique isosingular set of $f_\xi^\ab$
denoted by $\Iso_{f_{\xi}^\ab}(X\times \{0\})$ and $\Iso_{f_{\xi}^\ab}(Y\times \{0\})$,
respectively.  Let $F_i(x,t)$ and $F_j(x,t)$ be their corresponding deflated witness systems, respectively. If $F_i=F_j$, i.e. the two isosingular sets of $f_{\xi}^\ab$ are the same,
then $\Iso_{F_j(x,0)}(X)\neq \Iso_{F_j(x,0)}(Y)$ (otherwise $X=Y$) so $\by\in \Sing_{F_j(x,0)}(Y)$.  
Note that by the {\sc Deflated Witness System Algorithm} \ref{Alg:DWS},  $G_j(x)$ is the deflation of $F_j(x,0)$ at a generic point of $V_j$. This implies by Theorem \ref{thm:singdeflation} that $\by\in \Sing_{G_j}(Y)$ and $g_j(\by)=0$.  

If $F_i\neq F_j$, then $(\by,0)$ is in the intersection of two different isosingular sets so $(\by,0)$ has a different deflation sequence than a generic point in $Y\times \{0\}$, i.e.,  $(\by,0)\in \Sing_{f_{\xi}^\ab}(Y\times \{0\})$. By Theorem \ref{thm:singdeflation}, we have that $(\by,0)\in \Sing_{F_j}(Y\times \{0\})$. Denoting the Jacobian by $J:=JF_j(x,t)$, we have that $\rank J(\by,0)<s$ with 
$\rank J(\by',0)=s$ for all generic $\by'\in Y$. Consider $J':=JF_j(x,0)$. (i.e. column of $J$ corresponding to $\partial t$ removed). Note that $Jf(x)$ is a submatrix of $J'$, since $f\subset F_j(x,0)$.
If $\rank J'(\by')=s$ for  generic $\by'\in Y$, then $G_j=F_j(x,0)$, $\by \in \Sing_{G_j}(Y)$, and  $g_j(\by)=0$.  
If $\rank J'(\by')<s$  for  generic $\by'\in Y$, we claim that $\rank J'(\by)<\rank J'(\by')$  for  generic $\by'\in Y$. 
First note that both $\rank Jf(\by)\leq s-1$ and $ \rank Jf(\by')\leq s-1$ for $f=(f_1, \ldots, f_s)$, so without loss of generality (after maybe some Gaussian elimination on these Jacobian matrices), we assume that $\nabla f_1(y)=\nabla f_1(y')=0$.   Note that the $\partial t$ column of $J=JF_j(x,t)$   has the only possibly non-zero constant entries in the rows corresponding to $f_1-a_1t\xi,\ldots, f_s-a_st\xi$.  
Then for a generic $\by'\in Y$ we have $\rank J'(\by')=s-1$, since $\rank J(\by',0)=s$, thus  among all $s\times s$ minors of $J(\by',0)$ some has to be non-zero, and the only possible non-zeros are the ones that are $a_1$ times the $(s-1)\times(s-1)$ minors of $J'(\by')$, thus we must have  $a_1\neq 0$ and $\rank J'(\by')=s-1$. On the other hand, the $s\times s$ minors of $J(\by,0)$ contain all $(s-1)\times(s-1)$ minors of $J'(\by)$ times $a_1$, so all these minors of $J'(\by)$ must be zero. This implies that $\rank J'(\by)<s-1$.   Thus, $\rank J'(\by)<\rank J'(\by')$. 
In particular, $\by\in \Sing_{F_j(x,0)}(Y)$ and by Theorem \ref{thm:singdeflation},
$\by\in \Sing_{G_j}(Y)$ which implies that $g_j(\by)=0$.  This proves (vi), and the theorem.
\end{proof}

One advantage of the approach using isosingular deflation is that, in many problems, the number of iterations in the deflation process is a small constant (zero or one).  In this case, the degrees of the polynomials in the output of both {\sc Deflated Witness Set Algorithm \ref{Alg:DWS}} and {\sc Computation of g Algorithm \ref{Alg:Cg}} are comparable to the maximal degree of the input polynomials $f_1, \ldots, f_s$. 
On the other hand, the degree of the polynomial $\omega\in \R[x_1, \ldots, x_{n-s+1}]$ computed in the symbolic approach in \cite[Lemma 4.3]{SafYanZhi2018} mentioned at the beginning of this section is the degree of $V$ bounded by the product of the degrees of the input polynomials. 
Nonetheless, the disadvantage of our approach is that in the worst case, 
we need as many iterations in the deflation as the multiplicity of the points 
and this may result polynomials with higher degree than the degree of $\omega$ in \cite[Lemma 4.3]{SafYanZhi2018}. We have the following bound on the degree of $g$ as a function on the number of iterations in the deflation:

\begin{proposition}\label{prop:boundg}
Let $f=(f_1, \ldots, f_s)$ and $\ab=(a_1, \ldots, a_s)\in \R^s$ such that $
V_e^\ab:=\V(f_1-a_1e, \ldots, f_s-a_s e)$ satisfies Assumption {\tt (A)}. Let $D:=\max_{i=1}^s \{\deg(f_i)\}$ and fix  $\bp\in V:=\lim_{e\rightarrow 0} V_e^\ab$. If {\sc Algorithm \ref{Alg:DWS}} takes $k$ iterations of the isosingular deflation to output $G\subset \R[\bx]$, the degrees of the polynomials in $G$ are bounded by $s^kD$. Furthermore, if $g(\bx):=\det(M(\bx))\in \R[\bx]$ where $M(\bx)$ is a $s\times s$ submatrix of $JG(\bx)$, then $\deg(g)\leq s^{k+1}D$.
\end{proposition}

\begin{proof}
The first claim follows from the fact that each iteration of the deflation algorithm adds the minors of the Jacobian of the polynomials in the previous iteration, and these minors have size less than $s$. Thus, the degrees of polynomials added to the system in each iteration are  at most $s$ times the  degrees of the polynomials in the  previous iteration. The second claim follows from the first.
\end{proof}

\section{Finite Critical Points of g} 

In this section, we establish a key result characterizing when a function $g$ will have a finite number of critical points over an algebraic set. This is an adaptation of Theorem 36 and Lemma 37 from \cite{Hongetal2020}. 

\begin{definition} Given $f_1, \ldots, f_s, g\in \R[\bx]$. We say that $\bx\in \C^n$ is a {\em critical point} of $g$ for $\V(f_1, \ldots, f_s)$ if $\bx\in V(f_1, \ldots, f_s)$ and 
$$
\nabla g(\bx)\in {\rm span}_\C \left(\nabla f_1(\bx), \ldots, \nabla f_s(\bx)\right),
$$
where $\nabla$ denotes the gradient operation.
\end{definition}

We give the following example to illustrate the possibility of a $g$ having infinite critical points over a smooth $\V$, to motivate why the rest of this section is necessary. 

\begin{example}\label{exinfcritpoints}
Let $f = x^2+4y^2-4xy+2$ and $g=x$. Then there are an infinite number critical points of $g$ over the algebraic set $\V(f)$, defined by $(x,y) = \bigg(x, \frac{x}{2}$\bigg).
\end{example}

We need the following corollary of Sard's theorem from \cite[Theorem A.6.1]{SommeseWampler2005}. It uses the notion of {\em quasi-projective} sets, which are the intersections of a Zariski-open and a Zariski-closed subset inside some projective space. Let $X_{\rm reg}$ denote the set of smooth points in $X$.


\begin{theorem}\label{thm:Sard-complex}
Let $f(\bx)$ denote a system of $n$ algebraic functions on an irreducible quasiprojective set $X$. Then there is a Zariski openset $U\subset \overline{f(X)}\subset \C^n$ such that for $\by\in U$, $\V(f(\bx)-\by)\cap X_{\rm reg}$ is smooth of dimension equal to the corank of $f$, i.e. $\dim X-\dim \overline{f(X)}$. Moreover, the Jacobian matrix of $f$ is of rank equal to $\dim X-\dim \overline{f(X)}$ at all points of $\V(f(\bx)-\by)\cap X_{\rm reg}$. 
\end{theorem}

\begin{theorem}\label{thm:finite-crit}
Let  $f_1, \ldots, f_s\in \R[\bx]$ and assume that $\V(f_1, \ldots, f_s)\subset \C^n$ is a smooth equidimensional algebraic set of dimension $n-s$.  Let $g_0\in \R[\bx]$.  Then there exists a Zariski closed proper subset ${\mathcal S}$ of $\C^n$ with $\dim({\mathcal S})<n$ such that for all $c= (c_1, \ldots, c_n)\in \R^n\setminus {\mathcal S}$ the polynomial
$$
g:=g_0\cdot\left((x_1-c_1)^2+\cdots + (x_n-c_n)^2 +1\right)\in \R[\bx]
$$
has finitely many critical points for  $\V(f_1, \ldots, f_s)$ where $g$ does not vanish. 
\end{theorem}

\begin{proof}
Let $V$ be an irreducible component of $\V(f_1, \ldots, f_s)$. By our assumptions, $\dim(V)=n-s$ and $V$ is smooth. We will prove that $g$ has finitely many critical points for $\V(f_1, \ldots, f_s)$ that lie in $V\setminus \V(g)$, and since this will be true for all irreducible components of $\V(f_1, \ldots, f_s)$, we get the claim of the theorem. 

We can assume that 
$$\dim (V\cap \V(g))< n-s
$$ otherwise, since $V$ is irreducible, $V\subset \V(g)$ and there is nothing to prove.

To simplify the notation, define for $c=(c_1, \ldots, c_n)\in \R^n$
$$
U_c(\bx):=(x_1-c_1)^2+\cdots + (x_n-c_n)^2 +1.
$$
Then 
$$
\nabla g(\bx)= U_c(\bx)\nabla g_0(\bx) +g_0(\bx)\nabla U_c(\bx)
$$
Thus, a point $\bx\in V$ is a critical point of $g$ for $\V(f_1, \ldots, f_s)$ if and only if 
$$
U_c(\bx)\nabla g_0(\bx) +g_0(\bx)\nabla U_c(\bx)\;\in\; {\rm span}_\C \left(\nabla f_1(\bx), \ldots, \nabla f_s(\bx)\right),
$$ 
This implies that $\bx\in V$ is a critical point of $g$ for $\V(f_1, \ldots, f_s)$ such that $g(\bx)\neq 0$ if and only if there exists $\lambda=(\lambda_1, \ldots, \lambda_s)\in \C^s$ such that 
{\small\begin{eqnarray*}
\left[\begin{array}{c}
c_1\\\vdots\\
c_n
\end{array}\right]
=
\frac{U_c(\bx)}{g_0(x)}\left[\begin{array}{c}
\partial_{x_1}g_0(\bx)\\\vdots\\
\partial_{x_n}g_0(\bx)
\end{array}\right]
+2\left[\begin{array}{c}
x_1\\\vdots\\
x_n
\end{array}\right]-
\lambda_1 \left[\begin{array}{c}
\partial_{x_1}f_1(\bx)\\\vdots\\
\partial_{x_n}f_1(\bx)
\end{array}\right]
\cdots-
\lambda_s \left[\begin{array}{c}
\partial_{x_1}f_s(\bx)\\\vdots\\
\partial_{x_n}f_s(\bx)
\end{array}\right].
\end{eqnarray*}}

Let 
$$W:=\{(\bx, t, \lambda)\in V\times \C^{s+1}\;|\; g(\bx)\neq 0, t\neq 0\}$$ 
and define  $p_i:W\rightarrow\C$ for $i=1, \ldots, n$, 
$$
p_i(\bx, t, \lambda):=t\partial_{x_i}g_0(\bx)+2x_i-\lambda_1\partial_{x_i}f_1(\bx)-\cdots - 
\lambda_s\partial_{x_i}f_s(\bx).
$$

Thus, $\bx\in V\setminus \V(g)$ is a critical point of $g$ for $\V(f_1, \ldots, f_s)$ if and only if there exists $(t,\lambda)\in \C^{s+1}$ such that $(\bx,t,\lambda)$  satisfies  
$$t=\frac{U_c(\bx)}{g_0(\bx)} \;\;\text{ and }\;\;p_i(\bx,t,\lambda)=c_i\quad  i=1, \ldots, n.$$

First we prove  that for  $\bp=(p_1, \ldots, p_n): W\rightarrow \C^n$, $\bp$ is dominant.  For all $\bx^*\in V$ and for $t=\lambda_1=\cdots =\lambda_s=0$ we have
 $$J_Wp(\bx^*,0,0)=[2\cdot I_{n-s}| \nabla g_0(\bx^*)| -Jf(\bx^*)],
 $$
 where $J_Wp$ is the Jacobian  of $\bp$ in a local parametrization of $W$ at $(\bx^*,0,0)$. By our assumtion on $V$,  ${\rm rank} Jf(\bx^*)=s$, thus ${\rm rank} J_Wp(\bx^*,0,0)\geq n$. This implies that the image of $\bp$ is $n$-dimensional, thus $\bp$ is dominant. Since $W$ inherits the irreducibility and  and smoothness of $V$, we get that $\overline{\bp(W)}=\C^n$.
 
We can apply Theorem \ref{thm:Sard-complex}  for $\bp$, so there exists a Zariski closed subset ${\mathcal S}$ of $\C^n$ and  such that for all $c\in \C^n\setminus {\mathcal S}$
for $W_1:=\{(\bx, t, \lambda)\in W\;|\; \bp(\bx, t, \lambda)=c\}$ we have 
$$
\dim(W_1)= \dim(W) - n=1
$$ 
using that $\dim(W)=n-s+s+1=n+1$ by our assumption that  $\dim (V\cap \V(g))<n-s$.\\

 Fix $c\in \R^n\setminus {\mathcal S}$. Next we show that  that 
 $$\dim\left\{(\bx, t, \lambda)\in W_1\;:\;U_c(\bx)-tg_0(x)=0
 \right\}=0.
 $$
If the above dimension is not $0$ then $0$ is a critical value of the function  $q(\bx,t):=U_c(\bx)-tg_0(x)\;:\; W_1\rightarrow \C$.  If we have such a critical value, then there exists $(\bx^*,t^*, \lambda^*)\in W_1$ such that $\nabla q(\bx^*,t^*)=0$, i.e.
$$
\left[
\partial_{x_1}q(\bx^*,t^*),\ldots\\
\partial_{x_n}q(\bx^*,t^*), 
g_0(\bx^*)
\right]=
\left[
0,\ldots,0,
0
\right].
$$
Thus we must have $g_0(\bx^*)=0$. However, $W_1\subset W$, so for $(\bx^*,t^*, \lambda^*)\in W_1$ we have $g_0(\bx^*)\neq 0$,  a contradiction. 

This implies that for any $c\in \R^n\setminus {\mathcal S}$, the solution set of $p_i(\bx,t,\lambda)=c_i$ for $i=1, \ldots, n$ and the equation $t=\frac{U_c(\bx)}{g_0(\bx)}$ is a zero dimensional subset $Z\subset W$. The set $\{\bx : (\bx,t,\lambda)\in Z\}$ is the finite set of critical points of $g$ for $\V(f_1, \ldots, f_s)$ in $V\setminus \V(g)$. 

\end{proof}

\section{Computation of Real Smooth Points - General Case} \label{sec:smooth2}

After introducing all necessary theory and subroutines for our purposes, now we are ready to return to our main topic, computing smooth points on general real algebraic varieties. 

We first define two genericity assumptions, informed by our previous results, in particular Proposition \ref{prop:NN} and Theorem \ref{thm:finite-crit}. Recall that $\crit(V,\pi_i)$ is the polar variety of the algebraic set $V$ with respect to the projection $\pi_i$ as in Definition \ref{def:polar}.

\begin{definition}\label{def:BB}
Consider polynomial $F\in \R[x]$  with $\V(F) \cap \R^n$ bounded, matrix $A \in {\rm GL_n(\R)}$. Define $V^A = \V(F^A)$ and $V_e^A:=\V(F^A-e) \subset \C^n$ for some constant $e >0$. We say that $F$ and $A$ satisfy Assumption {\tt (B)} if:
\begin{enumerate}
\item[{\tt (1)}:] there exists $e_0 > 0$ such that for all $0<e\leq e_0$ and all $1 \leq i \leq n$, $\crit(V_e^A, \pi_i)$ is either empty or is smooth and equidimensional with complex dimension $i-1$;

\item[{\tt (2)}:] for all $\bp \in V^A \cap \R^n$, $\pi_d^{-1}(\pi_d(\bp)) \cap (V^A \cap \R^n)$ is finite, where $d$ is greater than or equal to the local real dimension of $V^A$ at $\bp$;

\end{enumerate}
\end{definition}

\begin{definition}\label{def:CC}
Consider polynomials $F, g \in \R[\bx]$  and constant $c=(c_1, \hdots, c_n) \in \R^n$. Define $V_e:=\V(F-e) \subset \C^n$ for some constant $e >0$. We say that $F,g$ and $c$ satisfy Assumption {\tt (C)} if:
\begin{enumerate}
\item[{\tt (C)}:] There exists $e_0 > 0$ such that for all $0<e\leq e_0$, all $1 \leq i \leq n$, the polynomial
$$
\overline{g}:=g\cdot\left((x_1-c_1)^2+\cdots + (x_n-c_n)^2 +1\right)\in \R[\bx]
$$
has finitely many critical points for the polar variety $\crit(V_e, \pi_i)$  where $g$ does not vanish. 

\end{enumerate}
\end{definition}

The following theorem and corresponding proof establish the correctness of the main algorithm of the paper, {\sc Real Smooth Point Algorithm \ref{alg:RSP2}}.

\begin{theorem}
\label{thm:perturbed}
Fix $n, i, f_1, \hdots, f_s$ as in the input of {\sc Algorithm \ref{alg:RSP2}} such that $V(f_1, \ldots, f_s)\cap \R^n$ is bounded. Assume $A \in GL_n(\R)$ such that $A$ and  $F = f_1^2 + \cdots + f_s^2$ satisfy Assumption {\tt (B)} as in Definition \ref{def:BB}.  Also, for each $j=1, \ldots, r$, in Step (4) of {\sc Algorithm \ref{alg:RSP2}} we assume that $F^A, g_j$ and $c$ satisfy Assumption {\tt (C)} as in Definition \ref{def:CC}.  Then {\sc Algorithm \ref{alg:RSP2}} is correct. Furthermore, let $Z$ the output of {\sc Algorithm \ref{alg:RSP2}}.  If $Z=\emptyset$, then $\V(f_1, 
\hdots, f_s) \cap \R^n$ has no connected components of dimension $i-1$. If $Z\neq \emptyset$, then $\V(f_1, \hdots, f_s)\cap \R^n $ has some connected components
 of dimension at least $ i-1$.
\end{theorem}

\begin{proof} Fix $1\leq i \leq n$. Using the notation $V^A_e:=V(F^A-e)$,   by Assumption {\tt (B)}, $\crit(V^A_e,\pi_i)$ is smooth and equidimensional of dimension $i-1$ for all sufficiently small $e>0$. Only locally in this proof we use the simplified notation 
$$V:=\lim_{e\rightarrow 0}\crit(V^A_e,\pi_i)\subset \C^n$$
without designating its dependence on $i$ or $A$, otherwise the indices would become too involved. By Proposition \ref{prop:connectedlimits} (over $\C$ instead of $\R$),
 the set $V$
  is a Zariski closed set 
that is either equidimensional of dimension~$i-1$
or empty. Assume that $\{(g_j, (G_j, L, W_j)): j=1,...,r\}$ satisfies output specifications (i)-(vi) of  Algorithm \ref{Alg:Cg}. 
Fix $j\in \{1, \ldots, r\}$ and let $V_j\subset V$ be the union of irreducible components of $V$ with witness set $(G_j,L,W_j)$.
First we establish that $U_j$ defined in Step (4a) is finite.
We note that  $\crit(V^A_e,\pi_i)$ is smooth and equidimensional for all sufficiently small $e>0$. If $|U_j|=\infty$ then we redefine $g_j$ with a generic $c\in \R^n$. Using Assumption {\tt (C)}, we get that the redefined $U_j$ is finite and the loop will terminate. 

Next, since $\dim (V_j\cap \V(g_j))<i-1$  by (v) in Algorithm \ref{Alg:Cg}, 
either $(V_j\setminus \V(g_j))\cap \R^n=\emptyset$ or for each bounded connected component $C$ of $V_j\cap \R^n$  where $g_j$ is not identically zero, there exists $\bz\in U_j\cap C$ such that $g_j(\bz)\neq 0$. 
  Suppose $(V_j\setminus \V(g_j))\cap \R^n\neq \emptyset$.  Let $C_1, \ldots,  C_t\subset V_j\cap \R^n$ be  the bounded connected components of $V_j\cap \R^n$ where $g_j$ is not identically zero. Fix $m\in \{1, \ldots, t\}$.  Since  each $C_m$ is compact, the  distance from $C_m$ to $C_k$ is positive for each $m\neq k$. 
Also, for all sufficiently small $e$,  $V^A_e\cap \R^n$ is also compact. 
Since $C_m\subset V\cap \R^n$ is compact, Proposition \ref{prop:connectedlimits} shows that there exist connected components  $C_{m,1}^{(e)}, \ldots,C_{m,s_m}^{(e)}$ of  $V^A_e\cap \R^n$ for all sufficiently small $e>0$  such that $C_m=\bigcup_{l=1}^{s_m}\lim_{e\rightarrow 0^+}C_{m,l}^{(e)}$, each~$C_{m,l}^{(e)}$ is bounded, and since $C_m$ and $C_j$ has positive distance for $m\neq j$, also by  Proposition \ref{prop:connectedlimits} we have that
$$\cup_{l=1}^{s_m}C_{m,l}^{(e)} \cap \cup_{l=1}^{s_j}C_{j,l}^{(e)} = \emptyset$$ for all $j\not = m$.  For each $l=1, \ldots, s_m$, let ${\mathcal S}_{m,l}^{(e)}:=\pi_x (\V(L^{(j)})) \cap C_{m,l}^{(e)}$. By Lemma \ref{lem:noperturb}, ${\mathcal S}_{m,l}^{(e)}\neq \emptyset$ and it contains all points in $C_{m,l}^{(e)}$ where $g_j$ takes its extreme values. Let ${\mathcal S}_m:=\bigcup_{l=1}^{s_m}\lim_{e\rightarrow 0} {\mathcal S}_{m,l}^{(e)}$. Since~${\mathcal S}_{m,l}^{(e)}$ is bounded for all sufficiently small $e$, none of the limit points escape to infinity.   
Suppose that for all $\bz\in {\mathcal S}_m$ we have $g_j(\bz)=0$. Since $C_m$ is compact, by the Extreme Value Theorem, $g_j$ attains both a minimum and a maximum on $C_m$.
Since $g_j$ is not identically zero on $C_m$, 
either the minimum or the maximum value of $g_j$ on $C_m$ must be nonzero.  Let 
$\bz^*\in C_m$ such that $|g_j(\bz^*)|>0$. Let $\bz_e^*\in C_{m,l}^{(e)}$ for some $l=1, \ldots, s_m$ such that $\lim_{e\rightarrow 0} \bz_e^*=\bz^* $. Then for any  $\bz\in {\mathcal S}_m$, if $\bz_e \in {\mathcal S}_m^{(e)}$ such that $\lim_{e\rightarrow 0} \bz_e=\bz$, then for sufficiently small $e$  we have that $|g_j(\bz_e^*)|> |g_j(\bz_e)|$ by $\lim_{e\rightarrow 0} g_j(\bz_e)=g_j(\bz)=0$. Since ${\mathcal S}_m$ is finite, we can choose a common $e_0$ value for all $\bz\in {\mathcal S}_m$ so that if $0<e<e_0$ then $|g_j(\bz_e^*)|> |g_j(\bz_e)|$ for all $\bz_e\in {\mathcal S}_m^{(e)}$. Thus, ${\mathcal S}_m^{(e)}$ could not contain all points of $C_{m, l}^{(e)}$ for $l=1, \ldots, s_i$ where $g_j$ takes its extreme values, a contradiction. This proves $\pi_x\left(\lim_{e\rightarrow 0}\V(L^{(j)}_e)\right) \cap C_m$  contains a point  $\bz\in C_m$ such that $g_j(\bz)\neq 0$, i.e. $T_j\cap C_m\neq \emptyset$.

Next, let $Z_j=T_j\cap V_j$ and $Z=\bigcup_{j=1}^r Z_j$ as in Steps (4) and (5). Since $V=\bigcup_{j=1}^r V_j$ and for each $j=1, \ldots, r$,  $\Sing(V_j)\subset \V(g_j)$,  $V_k\cap V_j\subset V(g_j)$ for all $k\neq j$
 by (iii)-(vi) in {\sc Algorithm \ref{Alg:Cg}}, these points are smooth in $V_j\cap \R^n$, and also smooth in $V \cap \R^n$. Thus if $Z\neq \emptyset$, by  Theorem~\ref{thm:realconvdim} and Proposition \ref{prop:limpolar},  $\V(f_1, \hdots, f_s) \cap \R^n$ must have dimension $\geq i-1$ connected components. Conversely, if $\V(f_1, \hdots, f_s)\cap \R^n=V(F)\cap \R^n$ has a bounded connected component of dimension $i-1$, then by Proposition \ref{prop:limpolar} we have $V(F)\cap \R^n= V$, so there exists $j\in\{1, \ldots, r\}$ such that  $V_j\cap \R^n$ has a bounded connected component of dimension $i-1$. By  Theorem \ref{thm:realconvdim}, this component has  real smooth points. In fact, these real smooth points form a semi-algebraic set that has also dimension $i-1$. 
However, since  $\dim \left(V_j\cap \V(g_j)\right)<i-1$, $g_j$ does not vanish on all real smooth points of this component, but it vanishes on the singular points. By the above argument $T_j\cap V_j $ must contain points where~$g_j$ is not zero, thus $Z_j$ and $Z$ are not empty.
\end{proof}

\begin{example}\label{ex:NotequiDim}
Consider $f_1,f_2\in\R[x_1,x_2,x_3]$ where
$$f_1 = (x^2+1)(x^2+y^2+z^2-1) \hbox{~~and~~}
f_2 = (x^2+1)(x+y+z-1).$$
Clearly, $V(f_1,f_2)$ is not equidimensional,
but $V(f_1,f_2)\cap\R^3$ is compact of dimension $1$.  
With $\ab=(1,1)$,   the limit variety $V$ is a curve
with two irreducible components: $V_1 = V(x^2+y^2+z^2-1,x+y+z-1)$
and $V_2 = V(x^2+1,x^2+y^2+z^2-x-y-z)$. We utilize 
{\sc Algorithm~\ref{alg:RSP2}} to compute a smooth point
on this real curve.
Using $g_1 = x-y$ and $g_2 = x(2y-1)$, 
respectively, one obtains 
$S_1 = \{(1\pm\sqrt{3},1\mp\sqrt{3},1)/3\}$ consisting
of two smooth points on $V_1\cap\R^3$ and $S_2 = \emptyset$.
\end{example}

Using Proposition \ref{prop:boundg}, we can bound the number of homotopy paths followed in Step (3) in the {\sc Real Smooth Point Algorithm~\ref{alg:RSP2}}, which is the bottleneck of our method. Note that the number of iterations $r$ is at most $\deg(V)\leq D^n$ where $D:=\max_{i=1}^s \{\deg(f_i)\}$. Thus the  {\sc Membership Test Algorithm \ref{Alg:MT}} utilized 
in Step (4) of the {\sc Real Smooth Point Algorithm~\ref{alg:RSP2}} follows at most $|W_j|=\deg(V_j)\leq \deg(V)\leq D^n$  homotopy paths.

\begin{algorithm}[H]
\caption[Alg:RSP2]{\sc RealSmoothPoint }
\label{alg:RSP2}
\begin{description}
\item[Input:] $f=\left(  f_{1},\ldots,f_{s}\right)  \subset\mathbb{R}\left[
x_{1},\ldots,x_{n}\right]  ,$ $i\in\left\{  1,\ldots,n\right\}  ,$ $n\geq2$.

\item[Output:] $Z\subset\mathbb{R}^{n}$, a finite set containing smooth points in
each $\left(  i-1\right)  $-dimensional bounded connected component of $\V_{\mathbb{R}%
}\left(  f\right)  $
\end{description}

\begin{enumerate}
\item Define $F:= f_{1}^{2}+\cdots+f_{s}^{2}$ and $e_1 :=(1,0,\hdots, 0)$.

\item Choose generic $A\in GL_{n}\left(  \mathbb{R}\right)  $.

\item $\left\{  \left(  g_{1},D_{1}\right)  ,\ldots,\left(  g_{r}%
,D_{r}\right)  \right\}$ := \textsf{ComputationOfG}$\left(
F^{A},\frac{\partial F^{A}}{\partial x_{i+1}},\ldots,\frac{\partial F^{A}}{\partial x_{n}},\textbf{e}_{1}\right)$. 
{\small {\tt // See Algorithm \ref{Alg:Cg}}}

{\small {\tt // $D_j$ is a deflated witness set for some $V_j$ a union of irreducible components of $\lim_{e \rightarrow 0^+}\crit(V_e^A, \pi_i)$ where $\crit(V_e^A,\pi_i):=\V(F^A-e, \frac{\partial F}{\partial x_{i+1}}, \hdots, \frac{\partial F}{\partial x_n})$ for $e$ a parameter.}}

\item For $j=1,\ldots,r$

\begin{enumerate}
\item Loop

$\qquad L^{\left(  j\right)  }:= \left\{F^{A},\frac{\partial F^A}{\partial
x_{i+1}},\ldots,\frac{\partial F^A}{\partial x_{n}}\right\} \cup \left\{  \frac{\partial g_{j}%
}{\partial x_{k}}+\lambda_0\frac{\partial F^{A}}{\partial x_{k}}+\sum\limits_{t=i+1}%
^{n}\lambda_{t}\frac{\partial^{2}F^A}{\partial x_{t}\partial x_{k}}%
:k=1,\ldots,n\right\}   $. \\

\qquad {\small {\tt // $L^{\left(  j\right)}$ is the Lagrange multiplier system in variables}} 

\qquad {\small {\tt $x_1, \ldots, x_n, \lambda_0,\lambda_1, \ldots, \lambda_s$.}}

\qquad$({\rm flag},U_{j})$ :=\textsf{WitnessPointsInLimits}$\left(  L^{\left(
j\right)  },\textbf{e}_{1}, \emptyset\right)  .$
\qquad {\small {\tt // See Algorithm \ref{Alg:WPL}}}

\qquad If ${\rm flag}=$TRUE, exit loop. 

\qquad Choose generic $c\in\mathbb{R}^{n}.$

\qquad$\overline{g_{j}}:= g_{j}\cdot\left(  \left(  x_{1}-c_{1}\right)
^{2}+\cdots\left(  x_{n}-c_{n}\right)  ^{2}+1\right)  .$

\qquad Restart loop with $g_j:=\overline{g_j}$.

\item Compute $T_{j}:= \pi_x(U_{j})\setminus V(g_j) \cap \R^n  .$ \qquad {\small {\tt // $\pi_x$ projection onto $x$ coordinates}}

\item Set $Z_{j}:=\emptyset.$

\item For each $\bp\in T_{j}$

\qquad If \textsf{MembershipTest}$\left(  \bp,D_{j}\right)  =$ TRUE, then 
\qquad {\small {\tt // See Algorithm \ref{Alg:MT}}}

$\qquad\qquad Z_{j}:= Z_{j}\cup\left\{ \bp\right\}  .$
\end{enumerate}

\item Return $Z:= \bigcup_{j=1}^{r}Z_{j}.$
\end{enumerate}
\end{algorithm}

\begin{corollary} \label{cor:boundLe} Fix $i\in \{1, \ldots, n\}$. Let $f_1, \ldots, f_s\in \R[\bx], A \in GL_n(\R)$ such that $A$ and $F = f_1^2 + \cdots + f_s^2$ satisfy Assumption {\tt (B)}. Assume that for some fixed $j\in \{1, \ldots, r\}$, the zero-dimensional polynomial system $L^{(j)}$  as in {\sc  Algorithm \ref{alg:RSP2}} Step (4a) is zero-dimensional. Then, the number of complex roots of~$L^{(j)}$ is bounded by 
$\deg(g_j)^n(2D)^{n-i+1}\leq (n-i+1)^{(k_j+1)n}(2D)^{2n-i+1},$
where $D:=\max_{i=1}^s \{\deg(f_i)\}$,  assuming that $\deg(g_j)\geq 2D$,  and $k_j$ is the number of iterations of the isosingular deflation needed to compute $G_j$ using {\sc Algorithm \ref{Alg:DWS}}.
\end{corollary}

\begin{proof}
In Step 3 we input $n-i+1$ polynomials of degrees at most $2D$ to Algorithm \ref{Alg:Cg}. By Proposition \ref{prop:boundg} with $s=n-i+1$, we get that $\deg(g_j)\leq (n-i+1)^{k_j+1}(2D)$ for all $j=1, \ldots, r$. Then the defining equations of $L^{(j)}$ include $n-i+1$ polynomials of degree at most $2D$ and $n$ polynomials of degree at most $\deg(g_j)$ (assuming that $\deg(g_j)\geq 2D$). The Bezout bound for the number of common roots of $L^{(j)}$ gives the claim of the Corollary.
\end{proof}

\section{Numerical Real Dimension Algorithm}

Our  real dimension algorithm is as follows.

\begin{algorithm}[ht]
\caption[Alg:NRD]{\sc NumericalRealDimension }
\label{Alg:NRD}
\begin{description}
\item[Input:] $f_1, \ldots, f_{s}\in \R[x_1, \hdots, x_n]$ not all zero, such that $\V(f_1, \ldots, f_s)\cap\R^n$ is bounded and $n\geq 2$.
\item[Output:] The real dimension of $\V(f_1, \ldots, f_s)\cap \R^n$.
\end{description}
 \begin{enumerate}
\item Let $i:=n$. 
\item Loop
\begin{enumerate}
\item $S:=\textsf{RealSmoothPoint}(f_1, \hdots, f_s, i)$. 
\qquad {\small {\tt // See Algorithm \ref{alg:RSP2}}}

{\small {\tt // $S\subset \R^n$ contains smooth points in  $\V(f_1, \hdots, f_s) \cap \R^n$. }}

\item If $S\neq \emptyset$, exit loop and return $i-1$.

\item \label{step:i-1} Increment $i:=i-1$. 
\item If $i=0$, exit loop and return $-1$.  
 \end{enumerate} 
 \end{enumerate}
\end{algorithm}

 \begin{theorem}\label{thm:NRDcorrect}
 Let $n\geq 2$, $f_1, \ldots, f_{s}\in \R[\bx]$ such that $\V(f_1, \ldots, f_s)\cap\R^n$ is bounded. Assume that the conditions of Theorem \ref{thm:perturbed} are satisfied for $1 \leq i \leq n$. Then {\sc Algorithm \ref{Alg:NRD}} is correct. 
 \end{theorem}

\begin{proof}
By assumption, Theorem \ref{thm:perturbed} gives the correctness of {\sc Real Smooth Point Algorithm \ref{alg:RSP2}} in Step (2a). We prove by induction on $n-i< n$ that we have the following loop invariant in Step (2):  $\dim (\V(f_1, \ldots, f_s)\cap \R^n)\leq i-1$. This is true when $n-i=0$. Assume it is true for $0\leq n-i<n$,  i.e. $\dim (\V(f_1, \ldots, f_s)\cap \R^n)\leq i-1$.   In Step (2b) if $S\neq \emptyset$ then by Theorem~\ref{thm:perturbed}, $V(f_1, \dots, f_s )\cap \R^n$ has some connected components of dimension at least $i-1$. By the inductive hypothesis, we get that this dimension must be equal to $i-1$, so that is the real dimension of $V(f_1, \dots, f_s )$ that we return and exit the loop. Otherwise, we have $S=\emptyset$.   We claim that in this case if $i>1$ then $\dim (\V(f_1, \ldots, f_s)\cap \R^n)< i-1$.  Since  $S=\emptyset$,  Theorem \ref{thm:perturbed} implies that there are no connected components of dimension $i-1$ in $V(f_1, \ldots, f_s)\cap\R^n$. Again, by the inductive hypothesis, we get that $\dim (\V(f_1, \ldots, f_s)\cap \R^n)< i-1$, maintaining the loop invariant. In particular,  if $i=1$ this implies that  $\V(f_1, \ldots, f_s)\cap \R^n=\emptyset$, i.e. the dimension is $-1$ by convension.  
\end{proof}

\begin{example}
The Whitney umbrella is a real algebraic set consisting  of a 2-dimensional umbrella-like surface with a 1-dimensional handle 
along the z-axis and defined by $f_1 = x^2-y^2z$. Since the surface is not compact, we add $f_2 = x^2+y^2+z^2+w^2-4$ following Proposition ~\ref{prop:bounded}. 
$g = x$ satisfies the requirements of Theorem~\ref{thm:noperturb}
and results in the red smooth points 
shown in Figure~\ref{fig:Whitney}(a),
confirming the real dimension is two.

To instead determine the local real dimension of the handle of the umbrella, we localize our computations by taking $f_2=x^2+y^2+(z+1)^2+w^2-\frac{1}{4}$. As expected, an optimization using $g = z+1$ results in no smooth real points. We return to Step 2 and add the equation defining our next polar variety, namely $f_3=2y^3-4yz^2-4yz$. 
Optimizing the new system with respect to $g=z+1$ 
yields the smooth real point on the handle
shown in Figure~\ref{fig:Whitney}(b), confirming the real dimension of the handle is one.
\end{example}

 \begin{figure}[!h]
    \centering
   \[ \begin{array}{cc}
    \includegraphics[scale=0.275]{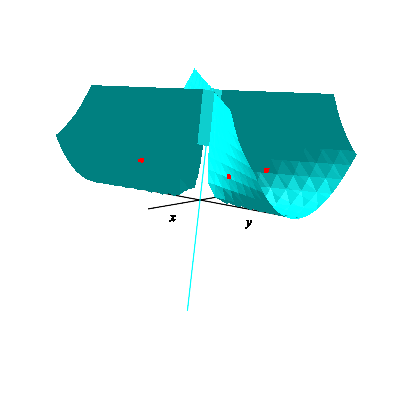} & \includegraphics[scale=0.275]{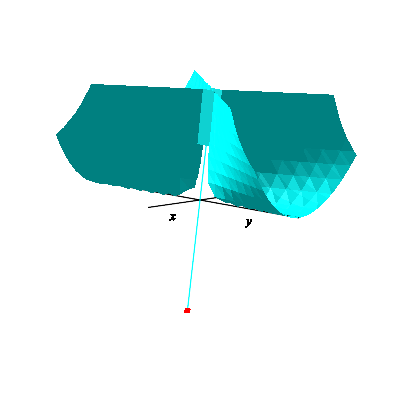} \\ 
    \text{(a) Dim. 2 Smooth Points}& \text{(b) Dim. 1 Smooth Point}
    \end{array}\]
    \caption{Whitney umbrella}
    \label{fig:Whitney}
\end{figure}

\section{Implementation on a Benchmark Family of Problems}

A benchmark family that appears in the papers  \cite{BanSaf2015} and \cite{LairezSafey2021} are hypersurfaces
$\V(f_n) \subset\C^n$ for $n\geq 3$ such that
\begin{equation}\label{eq:Benchmark}
\mbox{$f_n(\bx) = \left(\sum_{j=1}^n x_j^2\right)^2 - 4\sum_{j=1}^n \left(x_jx_{j+1}\right)^2$}
\end{equation}
where $x_{n+1}=x_1$.  Since $f_n$ is homogeneous, one knows $\dim \V(f_n)\cap\R^n = \dim (\V(f_n,s_n)\cap\R^n) + 1$ where~
\mbox{$s_n = \sum_{j=1}^n x_j^2-1$} 
in which $\V(f_n,s_n)\cap\R^n$ is compact.
The cases $3\leq n\leq 6$ were solved in \cite{BanSaf2015}, which were improved in \cite{LairezSafey2021} where they consider the cases $3\leq n\leq 8$.
Here we consider the cases  $3\leq n\leq 9$.
All code used in these computations is available at
\url{dx.doi.org/10.7274/r0-5c1t-jw53} with the timings
reported using {\tt Bertini} (\cite{Bertini})
on an AMD Opteron 6378 2.4 GHz processor
using one (serial) or 64 (parallel) cores.

For $n=3$ with $g = \partial f_3/\partial x_1$, one obtains
smooth points on $\V(f_3)\cap\R^3$ 
thereby showing $\dim \V(f_3)\cap\R^3 = 2$
in about a second in serial.

For $n=4$, $\V(f_4)$ has multiplicity $2$ with respect to $f_4$
since
$$f_4(x_1,x_2,x_3,x_4) = \left(x_1^2-x_2^2+x_3^2-x_4^2\right)^2.$$
Trivially, a deflated witness system for $\V(f_4)$ is $G = x_1^2-x_2^2+x_3^2-x_4^2$.  
For $g = x_1x_2$, one obtains smooth points on $\V(f_4)\cap\R^4$ showing
$\dim \V(f_4)\cap\R^4 = 3$ in about a second in serial.

For $n=5,\dots,9$, with $g = \partial f_n/\partial x_1$,
one does not obtain smooth points on $\V(f_n)\cap\R^n$
showing $\dim \V(f_n)\cap\R^n < n-1$.  Therefore, one can
move down the dimensions searching for real smooth points
using perturbed polar varieties, similarly to Step (2) of Algorithm~\ref{Alg:NRD}. Nonsingular real points are first found at dimension $2$, 
i.e., $\dim \V(f_n)\cap\R^n = 2$.
In fact, at dimension $2$, the polar variety contains 
various irreducible components of degree $2$ and testing 
one is enough to confirm the existence of a smooth real point.  
Table~\ref{tab:Benchmark2} lists the total computation
time using parallel~processing.
\begin{table}[!ht]
    \centering
    \caption{Summary of benchmark problem \eqref{eq:Benchmark} for $5\leq n\leq 9$}
    \begin{tabular}{c|c|r}
      $n$ & $\dim \V(f_n)\cap\R^n$ & Time (min)\\
      \hline
      $5$ & $2$ & $3.63$\hspace{.1in}~\\
      $6$ & $2$ & $5.73$\hspace{.1in}~\\
      $7$ & $2$ & $34.81$\hspace{.1in}~\\
      $8$ & $2$ & $159.81$\hspace{.1in}~\\
      $9$ & $2$ & $2675.25$\hspace{.1in}~\\
      \end{tabular}
    \label{tab:Benchmark2}
\end{table}

\section*{Acknowledgments}
This paper is dedicated to our friend and coauthor
Agnes Szanto who passed away during revisions.
The authors thank Mohab Safey El Din, Elias Tsigaridas, and Hoon Hong for many discussions
regarding real algebraic geometry
and the anonymous reviewers for their excellent comments and suggestions.
This research was partly supported by NSF  grants 
CCF-1812746~(Hauenstein) and CCF-1813340 (Szanto and Harris).

\bibliographystyle{elsarticle-harv}
\bibliography{References}

\begin{thebibliography}{51}
\expandafter\ifx\csname natexlab\endcsname\relax\def\natexlab#1{#1}\fi
\providecommand{\url}[1]{\texttt{#1}}
\providecommand{\href}[2]{#2}
\providecommand{\path}[1]{#1}
\providecommand{\DOIprefix}{doi:}
\providecommand{\ArXivprefix}{arXiv:}
\providecommand{\URLprefix}{URL: }
\providecommand{\Pubmedprefix}{pmid:}
\providecommand{\doi}[1]{\href{http://dx.doi.org/#1}{\path{#1}}}
\providecommand{\Pubmed}[1]{\href{pmid:#1}{\path{#1}}}
\providecommand{\bibinfo}[2]{#2}
\ifx\xfnm\relax \def\xfnm[#1]{\unskip,\space#1}\fi
\bibitem[{Aubry et~al.(2002)Aubry, Rouillier and Safey El~Din}]{AuRouSaf02}
\bibinfo{author}{Aubry, P.}, \bibinfo{author}{Rouillier, F.},
  \bibinfo{author}{Safey El~Din, M.}, \bibinfo{year}{2002}.
\newblock \bibinfo{title}{Real solving for positive dimensional systems}.
\newblock \bibinfo{journal}{J. Symbolic Comput.} \bibinfo{volume}{34},
  \bibinfo{pages}{543--560}.
\bibitem[{Bank et~al.(2015)Bank, Giusti, Heintz, Lecerf, Matera and
  Solernó}]{Banketal14}
\bibinfo{author}{Bank, B.}, \bibinfo{author}{Giusti, M.},
  \bibinfo{author}{Heintz, J.}, \bibinfo{author}{Lecerf, G.},
  \bibinfo{author}{Matera, G.}, \bibinfo{author}{Solernó, P.},
  \bibinfo{year}{2015}.
\newblock \bibinfo{title}{Degeneracy loci and polynomial equation solving}.
\newblock \bibinfo{journal}{Foundations of Computational Mathematics}
  \bibinfo{volume}{15}, \bibinfo{pages}{159–184}.
\bibitem[{Bank et~al.(1997)Bank, Giusti, Heintz, Mandel and
  Mbakop}]{Banketal1997}
\bibinfo{author}{Bank, B.}, \bibinfo{author}{Giusti, M.},
  \bibinfo{author}{Heintz, J.}, \bibinfo{author}{Mandel, R.},
  \bibinfo{author}{Mbakop, G.M.}, \bibinfo{year}{1997}.
\newblock \bibinfo{title}{Polar varieties and efficient real equation solving:
  the hypersurface case}, in: \bibinfo{booktitle}{Proceedings of the 3rd
  {I}nternational {C}onference on {A}pproximation and {O}ptimization in the
  {C}aribbean ({P}uebla, 1995)}, \bibinfo{publisher}{Benem\'{e}rita Univ.
  Aut\'{o}n. Puebla, Puebla}. p.~\bibinfo{pages}{13}.
\bibitem[{Bank et~al.(2004)Bank, Giusti, Heintz and Pardo}]{Banketal04}
\bibinfo{author}{Bank, B.}, \bibinfo{author}{Giusti, M.},
  \bibinfo{author}{Heintz, J.}, \bibinfo{author}{Pardo, L.M.},
  \bibinfo{year}{2004}.
\newblock \bibinfo{title}{Generalized polar varieties and an efficient real
  elimination procedure}.
\newblock \bibinfo{journal}{Kybernetika (Prague)} \bibinfo{volume}{40},
  \bibinfo{pages}{519--550}.
\bibitem[{Bank et~al.(2009)Bank, Giusti, Heintz and Pardo}]{Banketal09}
\bibinfo{author}{Bank, B.}, \bibinfo{author}{Giusti, M.},
  \bibinfo{author}{Heintz, J.}, \bibinfo{author}{Pardo, L.M.},
  \bibinfo{year}{2009}.
\newblock \bibinfo{title}{On the intrinsic complexity of point finding in real
  singular hypersurfaces}.
\newblock \bibinfo{journal}{Inform. Process. Lett.} \bibinfo{volume}{109},
  \bibinfo{pages}{1141--1144}.
\bibitem[{Bank et~al.(2010)Bank, Giusti, Heintz, Safey El~Din and
  Schost}]{Banketal2010}
\bibinfo{author}{Bank, B.}, \bibinfo{author}{Giusti, M.},
  \bibinfo{author}{Heintz, J.}, \bibinfo{author}{Safey El~Din, M.},
  \bibinfo{author}{Schost, E.}, \bibinfo{year}{2010}.
\newblock \bibinfo{title}{On the geometry of polar varieties}.
\newblock \bibinfo{journal}{Applicable Algebra in Engineering, Communication
  and Computing} \bibinfo{volume}{21}, \bibinfo{pages}{33--83}.
\bibitem[{Bannwarth and Safey El~Din(2015)}]{BanSaf2015}
\bibinfo{author}{Bannwarth, I.}, \bibinfo{author}{Safey El~Din, M.},
  \bibinfo{year}{2015}.
\newblock \bibinfo{title}{Probabilistic algorithm for computing the dimension
  of real algebraic sets}, in: \bibinfo{booktitle}{I{SSAC}'15---{P}roceedings
  of the 2015 {ACM} {I}nternational {S}ymposium on {S}ymbolic and {A}lgebraic
  {C}omputation}. \bibinfo{publisher}{ACM, New York}, pp.
  \bibinfo{pages}{37--44}.
\bibitem[{Basu et~al.(2006a)Basu, Pollack and Roy}]{BasuPollackRoyBook}
\bibinfo{author}{Basu, S.}, \bibinfo{author}{Pollack, R.},
  \bibinfo{author}{Roy, M.F.}, \bibinfo{year}{2006}a.
\newblock \bibinfo{title}{Algorithms in real algebraic geometry}.
  volume~\bibinfo{volume}{10} of \textit{\bibinfo{series}{Algorithms and
  Computation in Mathematics}}.
\newblock \bibinfo{edition}{Second} ed., \bibinfo{publisher}{Springer-Verlag,
  Berlin}.
\bibitem[{Basu et~al.(2006b)Basu, Pollack and Roy}]{BasuPolackRoy2006}
\bibinfo{author}{Basu, S.}, \bibinfo{author}{Pollack, R.},
  \bibinfo{author}{Roy, M.F.}, \bibinfo{year}{2006}b.
\newblock \bibinfo{title}{Computing the dimension of a semi-algebraic set}.
\newblock \bibinfo{journal}{Journal of Mathematical Sciences}
  \bibinfo{volume}{134}, \bibinfo{pages}{2346--2353}.
\bibitem[{Bates et~al.()Bates, Hauenstein, Sommese and Wampler}]{Bertini}
\bibinfo{author}{Bates, D.J.}, \bibinfo{author}{Hauenstein, J.D.},
  \bibinfo{author}{Sommese, A.J.}, \bibinfo{author}{Wampler, C.W.}, .
\newblock \bibinfo{title}{Bertini: Software for numerical algebraic geometry}.
\newblock \bibinfo{howpublished}{Available at \url{bertini.nd.edu}}.
\bibitem[{Bates et~al.(2013)Bates, Hauenstein, Sommese and
  Wampler}]{BertiniBook2013}
\bibinfo{author}{Bates, D.J.}, \bibinfo{author}{Hauenstein, J.D.},
  \bibinfo{author}{Sommese, A.J.}, \bibinfo{author}{Wampler, C.W.},
  \bibinfo{year}{2013}.
\newblock \bibinfo{title}{Numerically solving polynomial systems with
  {B}ertini}. volume~\bibinfo{volume}{25} of \textit{\bibinfo{series}{Software,
  Environments, and Tools}}.
\newblock \bibinfo{publisher}{Society for Industrial and Applied Mathematics
  (SIAM), Philadelphia, PA}.
\bibitem[{Becker and Neuhaus(1993)}]{BeckNeu93}
\bibinfo{author}{Becker, E.}, \bibinfo{author}{Neuhaus, R.},
  \bibinfo{year}{1993}.
\newblock \bibinfo{title}{Computation of real radicals of polynomial ideals},
  in: \bibinfo{booktitle}{Computational algebraic geometry ({N}ice, 1992)}.
  \bibinfo{publisher}{Birkh\"{a}user Boston, Boston, MA}. volume
  \bibinfo{volume}{109} of \textit{\bibinfo{series}{Progr. Math.}}, pp.
  \bibinfo{pages}{1--20}.
\bibitem[{Bernardi et~al.(2018)Bernardi, Blekherman and
  Ottaviani}]{RealTypicalRanks2018}
\bibinfo{author}{Bernardi, A.}, \bibinfo{author}{Blekherman, G.},
  \bibinfo{author}{Ottaviani, G.}, \bibinfo{year}{2018}.
\newblock \bibinfo{title}{On real typical ranks}.
\newblock \bibinfo{journal}{Boll. Unione Mat. Ital.} \bibinfo{volume}{11},
  \bibinfo{pages}{293--307}.
\bibitem[{Chen et~al.(2013)Chen, Davenport, May, Moreno~Maza, Xia and
  Xiao}]{Chenetal2013}
\bibinfo{author}{Chen, C.}, \bibinfo{author}{Davenport, J.H.},
  \bibinfo{author}{May, J.P.}, \bibinfo{author}{Moreno~Maza, M.},
  \bibinfo{author}{Xia, B.}, \bibinfo{author}{Xiao, R.}, \bibinfo{year}{2013}.
\newblock \bibinfo{title}{Triangular decomposition of semi-algebraic systems}.
\newblock \bibinfo{journal}{J. Symbolic Comput.} \bibinfo{volume}{49},
  \bibinfo{pages}{3--26}.
\bibitem[{Collins(1975)}]{Collins}
\bibinfo{author}{Collins, G.E.}, \bibinfo{year}{1975}.
\newblock \bibinfo{title}{Quantifier elimination for real closed fields by
  cylindrical algebraic decompostion}, in: \bibinfo{editor}{Brakhage, H.}
  (Ed.), \bibinfo{booktitle}{Automata Theory and Formal Languages},
  \bibinfo{publisher}{Springer Berlin Heidelberg}, \bibinfo{address}{Berlin,
  Heidelberg}. pp. \bibinfo{pages}{134--183}.
\bibitem[{Coss et~al.(2018)Coss, Hauenstein, Hong and Molzahn}]{Owenetal2018}
\bibinfo{author}{Coss, O.}, \bibinfo{author}{Hauenstein, J.D.},
  \bibinfo{author}{Hong, H.}, \bibinfo{author}{Molzahn, D.K.},
  \bibinfo{year}{2018}.
\newblock \bibinfo{title}{Locating and counting equilibria of the {K}uramoto
  model with rank-one coupling}.
\newblock \bibinfo{journal}{SIAM J. Appl. Algebra Geom.} \bibinfo{volume}{2},
  \bibinfo{pages}{45--71}.
\bibitem[{Elliott and Schost(2019)}]{EllSch2019}
\bibinfo{author}{Elliott, J.}, \bibinfo{author}{Schost, {\'E}.},
  \bibinfo{year}{2019}.
\newblock \bibinfo{title}{Bit complexity for critical point computation in
  smooth and compact real hypersurfaces}.
\newblock \bibinfo{journal}{ACM Communications in Computer Algebra}
  \bibinfo{volume}{53}, \bibinfo{pages}{114--117}.
\bibitem[{Faug\`{e}re et~al.(2008)Faug\`{e}re, Moroz, Rouillier and Safey
  El~Din}]{Faugereetal2008}
\bibinfo{author}{Faug\`{e}re, J.C.}, \bibinfo{author}{Moroz, G.},
  \bibinfo{author}{Rouillier, F.}, \bibinfo{author}{Safey El~Din, M.},
  \bibinfo{year}{2008}.
\newblock \bibinfo{title}{Classification of the perspective-three-point
  problem, discriminant variety and real solving polynomial systems of
  inequalities}, in: \bibinfo{booktitle}{ISSAC'08 -- Proceedings of the 2008
  ACM International Symposium on Symbolic and Algebraic Computation},
  \bibinfo{publisher}{ACM}, \bibinfo{address}{New York, NY, USA}. pp.
  \bibinfo{pages}{79--86}.
\bibitem[{Fischer(2001)}]{Fischer2001}
\bibinfo{author}{Fischer, G.}, \bibinfo{year}{2001}.
\newblock \bibinfo{title}{Plane algebraic curves}. volume~\bibinfo{volume}{15}
  of \textit{\bibinfo{series}{Student Mathematical Library}}.
\newblock \bibinfo{publisher}{American Mathematical Society, Providence, RI}.
\newblock \bibinfo{note}{Translated from the 1994 German original by Leslie
  Kay}.
\bibitem[{Friedland(2012)}]{TypicalRank2012}
\bibinfo{author}{Friedland, S.}, \bibinfo{year}{2012}.
\newblock \bibinfo{title}{On the generic and typical ranks of 3-tensors}.
\newblock \bibinfo{journal}{Linear Algebra and its Applications}
  \bibinfo{volume}{436}, \bibinfo{pages}{478--497}.
\bibitem[{Grayson and Stillman()}]{Macaulay2}
\bibinfo{author}{Grayson, D.R.}, \bibinfo{author}{Stillman, M.E.}, .
\newblock \bibinfo{title}{Macaulay2, a software system for research in
  algebraic geometry}.
\newblock \bibinfo{howpublished}{Available at
  \url{http://www.math.uiuc.edu/Macaulay2/}}.
\bibitem[{Grigor\cprime~ev and Vorobjov(1988)}]{GriVor88}
\bibinfo{author}{Grigor\cprime~ev, D.Y.}, \bibinfo{author}{Vorobjov, Jr.,
  N.N.}, \bibinfo{year}{1988}.
\newblock \bibinfo{title}{Solving systems of polynomial inequalities in
  subexponential time}.
\newblock \bibinfo{journal}{J. Symbolic Comput.} \bibinfo{volume}{5},
  \bibinfo{pages}{37--64}.
\bibitem[{Harris(2021)}]{HarrisThesis2021}
\bibinfo{author}{Harris, K.E.}, \bibinfo{year}{2021}.
\newblock \bibinfo{title}{{Real Algebraic Geometry with Numerical Homotopy
  Methods}}.
\newblock Ph.D. thesis. North Carolina State University.
\bibitem[{Hauenstein(2013)}]{Hauenstein2013}
\bibinfo{author}{Hauenstein, J.D.}, \bibinfo{year}{2013}.
\newblock \bibinfo{title}{Numerically computing real points on algebraic sets}.
\newblock \bibinfo{journal}{Acta Applicandae Mathematicae}
  \bibinfo{volume}{125}, \bibinfo{pages}{105--119}.
\bibitem[{Hauenstein and Sottile(2012)}]{HauSot2012}
\bibinfo{author}{Hauenstein, J.D.}, \bibinfo{author}{Sottile, F.},
  \bibinfo{year}{2012}.
\newblock \bibinfo{title}{Algorithm 921: alpha{C}ertified: certifying solutions
  to polynomial systems}.
\newblock \bibinfo{journal}{ACM Trans. Math. Software} \bibinfo{volume}{38},
  \bibinfo{pages}{Art. 28, 20}.
\bibitem[{Hauenstein and Wampler(2013)}]{HauensteinWampler2013}
\bibinfo{author}{Hauenstein, J.D.}, \bibinfo{author}{Wampler, C.W.},
  \bibinfo{year}{2013}.
\newblock \bibinfo{title}{Isosingular sets and deflation}.
\newblock \bibinfo{journal}{Found. Comput. Math.} \bibinfo{volume}{13},
  \bibinfo{pages}{371--403}.
\bibitem[{Hauenstein and Wampler(2017)}]{HauensteinWampler2017}
\bibinfo{author}{Hauenstein, J.D.}, \bibinfo{author}{Wampler, C.W.},
  \bibinfo{year}{2017}.
\newblock \bibinfo{title}{Unification and extension of intersection algorithms
  in numerical algebraic geometry}.
\newblock \bibinfo{journal}{Appl. Math. Comput.} \bibinfo{volume}{293},
  \bibinfo{pages}{226--243}.
\bibitem[{Hong et~al.(2020)Hong, Rohal, Din and Schost}]{Hongetal2020}
\bibinfo{author}{Hong, H.}, \bibinfo{author}{Rohal, J.}, \bibinfo{author}{Din,
  M.S.E.}, \bibinfo{author}{Schost, E.}, \bibinfo{year}{2020}.
\newblock \bibinfo{title}{Connectivity in semi-algebraic sets i}.
\newblock \href{http://arxiv.org/abs/2011.02162}{{\tt arXiv:2011.02162}}.
\bibitem[{Krick et~al.(2001)Krick, Pardo and Sombra}]{KrPaSo2001}
\bibinfo{author}{Krick, T.}, \bibinfo{author}{Pardo, L.M.},
  \bibinfo{author}{Sombra, M.}, \bibinfo{year}{2001}.
\newblock \bibinfo{title}{Sharp estimates for the arithmetic
  {N}ullstellensatz}.
\newblock \bibinfo{journal}{Duke Math. J.} \bibinfo{volume}{109},
  \bibinfo{pages}{521--598}.
\bibitem[{Kruskal(1989)}]{Kruskal1989}
\bibinfo{author}{Kruskal, J.B.}, \bibinfo{year}{1989}.
\newblock \bibinfo{title}{Rank, decomposition, and uniqueness for {$3$}-way and
  {$N$}-way arrays}, in: \bibinfo{booktitle}{Multiway data analysis ({R}ome,
  1988)}. \bibinfo{publisher}{North-Holland, Amsterdam}, pp.
  \bibinfo{pages}{7--18}.
\bibitem[{Kuramoto(1975)}]{Kuramoto1975}
\bibinfo{author}{Kuramoto, Y.}, \bibinfo{year}{1975}.
\newblock \bibinfo{title}{{Self-entrainment of a population of coupled
  non-linear oscillators}}.
\newblock \bibinfo{journal}{Lect. Notes Phys.} \bibinfo{volume}{39},
  \bibinfo{pages}{420--422}.
\bibitem[{Lairez and Safey El~Din(2021)}]{LairezSafey2021}
\bibinfo{author}{Lairez, P.}, \bibinfo{author}{Safey El~Din, M.},
  \bibinfo{year}{2021}.
\newblock \bibinfo{title}{Computing the dimension of real algebraic sets}, in:
  \bibinfo{booktitle}{{ISSAC} 2021 - 46th {I}nternational {S}ymposium on
  {S}ymbolic and {A}lgebraic {C}omputation}, \bibinfo{publisher}{ACM, New
  York}. pp. \bibinfo{pages}{257--264}.
\bibitem[{Logar(1989)}]{Logar}
\bibinfo{author}{Logar, A.}, \bibinfo{year}{1989}.
\newblock \bibinfo{title}{A computational proof of the {N}oether normalization
  lemma}.
\newblock \bibinfo{journal}{Applied Algebra, Algebraic Algorithms and
  Error-Correcting Codes} ,
  \bibinfo{pages}{259–273}\DOIprefix\doi{10.1007/3-540-51083-4_65}.
\bibitem[{Ma et~al.(2016)Ma, Wang and Zhi}]{Maetal2016}
\bibinfo{author}{Ma, Y.}, \bibinfo{author}{Wang, C.}, \bibinfo{author}{Zhi,
  L.}, \bibinfo{year}{2016}.
\newblock \bibinfo{title}{A certificate for semidefinite relaxations in
  computing positive-dimensional real radical ideals}.
\newblock \bibinfo{journal}{J. Symbolic Comput.} \bibinfo{volume}{72},
  \bibinfo{pages}{1--20}.
\bibitem[{Marshall(2008)}]{Marshall2008}
\bibinfo{author}{Marshall, M.}, \bibinfo{year}{2008}.
\newblock \bibinfo{title}{Positive polynomials and sums of squares}. volume
  \bibinfo{volume}{146} of \textit{\bibinfo{series}{Mathematical Surveys and
  Monographs}}.
\newblock \bibinfo{publisher}{American Mathematical Society, Providence, RI}.
\bibitem[{Neuhaus(1998)}]{Neuhaus98}
\bibinfo{author}{Neuhaus, R.}, \bibinfo{year}{1998}.
\newblock \bibinfo{title}{Computation of real radicals of polynomial ideals.
  {II}}.
\newblock \bibinfo{journal}{J. Pure Appl. Algebra} \bibinfo{volume}{124},
  \bibinfo{pages}{261--280}.
\bibitem[{Rouillier et~al.(2000)Rouillier, Roy and Safey
  El~Din}]{RouRoySaf2000}
\bibinfo{author}{Rouillier, F.}, \bibinfo{author}{Roy, M.F.},
  \bibinfo{author}{Safey El~Din, M.}, \bibinfo{year}{2000}.
\newblock \bibinfo{title}{Finding at least one point in each connected
  component of a real algebraic set defined by a single equation}.
\newblock \bibinfo{journal}{Journal of Complexity} \bibinfo{volume}{16},
  \bibinfo{pages}{716 -- 750}.
\bibitem[{Safey El~Din(2007)}]{Safey2007}
\bibinfo{author}{Safey El~Din, M.}, \bibinfo{year}{2007}.
\newblock \bibinfo{title}{Testing sign conditions on a multivariate polynomial
  and applications}.
\newblock \bibinfo{journal}{Math. Comput. Sci.} \bibinfo{volume}{1},
  \bibinfo{pages}{177--207}.
\bibitem[{Safey El~Din and Schost(2003)}]{SafeySchost03}
\bibinfo{author}{Safey El~Din, M.}, \bibinfo{author}{Schost, E.},
  \bibinfo{year}{2003}.
\newblock \bibinfo{title}{Polar varieties and computation of one point in each
  connected component of a smooth algebraic set}, in:
  \bibinfo{booktitle}{Proceedings of the 2003 {I}nternational {S}ymposium on
  {S}ymbolic and {A}lgebraic {C}omputation}, \bibinfo{publisher}{ACM, New
  York}. pp. \bibinfo{pages}{224--231}.
\bibitem[{Safey El~Din and Spaenlehauer(2016)}]{SafeySpaen2016}
\bibinfo{author}{Safey El~Din, M.}, \bibinfo{author}{Spaenlehauer, P.J.},
  \bibinfo{year}{2016}.
\newblock \bibinfo{title}{Critical point computations on smooth varieties:
  degree and complexity bounds}, in: \bibinfo{booktitle}{Proceedings of the
  2016 {ACM} {I}nternational {S}ymposium on {S}ymbolic and {A}lgebraic
  {C}omputation}, \bibinfo{publisher}{ACM, New York}. pp.
  \bibinfo{pages}{183--190}.
\bibitem[{Safey El~Din and Tsigaridas(2013)}]{MohabElias2013}
\bibinfo{author}{Safey El~Din, M.}, \bibinfo{author}{Tsigaridas, E.},
  \bibinfo{year}{2013}.
\newblock \bibinfo{title}{A probabilistic algorithm to compute the real
  dimension of a semi-algebraic set}.
\newblock \bibinfo{journal}{CoRR} \bibinfo{volume}{abs/1304.1928}.
\newblock \URLprefix \url{http://arxiv.org/abs/1304.1928},
  \href{http://arxiv.org/abs/1304.1928}{{\tt arXiv:1304.1928}}.
\bibitem[{Safey El~Din and Tsigaridas(2018)}]{SafEl2018}
\bibinfo{author}{Safey El~Din, M.}, \bibinfo{author}{Tsigaridas, E.},
  \bibinfo{year}{2018}.
\newblock \bibinfo{title}{Personal communication}.
\newblock \bibinfo{howpublished}{ICERM - Nonlinear Algebra Program}.
\bibitem[{Safey El~Din et~al.(2018)Safey El~Din, Yang and Zhi}]{SafYanZhi2018}
\bibinfo{author}{Safey El~Din, M.}, \bibinfo{author}{Yang, Z.H.},
  \bibinfo{author}{Zhi, L.}, \bibinfo{year}{2018}.
\newblock \bibinfo{title}{On the complexity of computing real radicals of
  polynomial systems}, in: \bibinfo{booktitle}{I{SSAC}'18---{P}roceedings of
  the 2018 {ACM} {I}nternational {S}ymposium on {S}ymbolic and {A}lgebraic
  {C}omputation}. \bibinfo{publisher}{ACM, New York}, pp.
  \bibinfo{pages}{351--358}.
\bibitem[{Safey El~Din et~al.(2019)Safey El~Din, Yang and Zhi}]{SafYanZhi2019}
\bibinfo{author}{Safey El~Din, M.}, \bibinfo{author}{Yang, Z.H.},
  \bibinfo{author}{Zhi, L.}, \bibinfo{year}{2019}.
\newblock \bibinfo{title}{Computing real radicals and {S-radicals} of
  polynomial systems}.
\newblock \bibinfo{journal}{Journal of Symbolic Computation} .
\bibitem[{Sommese and Wampler(2005)}]{SommeseWampler2005}
\bibinfo{author}{Sommese, A.J.}, \bibinfo{author}{Wampler, II, C.W.},
  \bibinfo{year}{2005}.
\newblock \bibinfo{title}{The numerical solution of systems of polynomials}.
\newblock \bibinfo{publisher}{World Scientific Publishing Co. Pte. Ltd.,
  Hackensack, NJ}.
\newblock \bibinfo{note}{Arising in engineering and science}.
\bibitem[{Spang(2008)}]{Spang08}
\bibinfo{author}{Spang, S.J.}, \bibinfo{year}{2008}.
\newblock \bibinfo{title}{A zero-dimensional approach to compute real
  radicals}.
\newblock \bibinfo{journal}{Comput. Sci. J. Moldova} \bibinfo{volume}{16},
  \bibinfo{pages}{64--92}.
\bibitem[{Vorobjov(1999)}]{Vorobjov1999}
\bibinfo{author}{Vorobjov, N.}, \bibinfo{year}{1999}.
\newblock \bibinfo{title}{Complexity of computing the local dimension of a
  semialgebraic set}.
\newblock \bibinfo{journal}{J. Symbolic Comput.} \bibinfo{volume}{27},
  \bibinfo{pages}{565--579}.
\bibitem[{Wang(2016)}]{Wang2016}
\bibinfo{author}{Wang, F.}, \bibinfo{year}{2016}.
\newblock \bibinfo{title}{Computation of Real Radical Ideals by Semidefinite
  Programming and Iterative Methods}.
\newblock Ph.D. thesis. University of Western Ontario.
\bibitem[{Wu and Reid(2013)}]{WuReid13}
\bibinfo{author}{Wu, W.}, \bibinfo{author}{Reid, G.}, \bibinfo{year}{2013}.
\newblock \bibinfo{title}{Finding points on real solution components and
  applications to differential polynomial systems}, in:
  \bibinfo{booktitle}{I{SSAC} 2013---{P}roceedings of the 38th {I}nternational
  {S}ymposium on {S}ymbolic and {A}lgebraic {C}omputation},
  \bibinfo{publisher}{ACM, New York}. pp. \bibinfo{pages}{339--346}.
\bibitem[{{Xin} et~al.(2016){Xin}, {Kikkawa} and {Liu}}]{KuramotoConjecture}
\bibinfo{author}{{Xin}, X.}, \bibinfo{author}{{Kikkawa}, T.},
  \bibinfo{author}{{Liu}, Y.}, \bibinfo{year}{2016}.
\newblock \bibinfo{title}{Analytical solutions of equilibrium points of the
  standard {Kuramoto} model: 3 and 4 oscillators}, in: \bibinfo{booktitle}{2016
  American Control Conference (ACC)}, pp. \bibinfo{pages}{2447--2452}.
\bibitem[{Zeng(1999)}]{Zeng1999}
\bibinfo{author}{Zeng, G.}, \bibinfo{year}{1999}.
\newblock \bibinfo{title}{Computation of generalized real radicals of
  polynomial ideals}.
\newblock \bibinfo{journal}{Sci. China Ser. A} \bibinfo{volume}{42},
  \bibinfo{pages}{272--280}.

\end{thebibliography}

\end{document}